\documentclass[11pt,a4paper]{article}

\usepackage[T1]{fontenc}
\usepackage[utf8]{inputenc}
\usepackage{lmodern}
\usepackage[margin=25mm,headheight=15pt,headsep=8mm]{geometry}
\usepackage{amsmath,amssymb,amsthm,mathtools}
\usepackage{microtype}
\usepackage{needspace}
\usepackage{graphicx}
\usepackage{booktabs,array,tabularx}
\usepackage{enumitem}
\usepackage{xcolor}
\usepackage{tikz}
\usetikzlibrary{arrows.meta,positioning}
\usepackage{titlesec}
\usepackage{aliascnt}
\usepackage{fancyhdr}
\usepackage[authoryear,round]{natbib}
\usepackage[hyphens]{url}
\usepackage[
  colorlinks=true,
  linkcolor=black,
  citecolor=black,
  urlcolor=black,
  bookmarksnumbered=true,
  pdfencoding=auto,
  psdextra
]{hyperref}
\usepackage[nameinlink,noabbrev]{cleveref}

\hypersetup{
  pdftitle={Models of Strategic Choice},
  pdfauthor={},
  pdfsubject={Working paper},
  pdfkeywords={game theory, extensive-form games, equilibrium, deliberation, behavioral rules}
}

\newtheorem{theorem}{Theorem}[section]

\newaliascnt{proposition}{theorem}
\newtheorem{proposition}[proposition]{Proposition}
\aliascntresetthe{proposition}

\newaliascnt{lemma}{theorem}
\newtheorem{lemma}[lemma]{Lemma}
\aliascntresetthe{lemma}

\newaliascnt{corollary}{theorem}
\newtheorem{corollary}[corollary]{Corollary}
\aliascntresetthe{corollary}

\theoremstyle{definition}

\newaliascnt{definition}{theorem}
\newtheorem{definition}[definition]{Definition}
\aliascntresetthe{definition}

\newaliascnt{assumption}{theorem}
\newtheorem{assumption}[assumption]{Assumption}
\aliascntresetthe{assumption}

\newaliascnt{example}{theorem}
\newtheorem{example}[example]{Example}
\aliascntresetthe{example}

\theoremstyle{remark}

\newaliascnt{remark}{theorem}
\newtheorem{remark}[remark]{Remark}
\aliascntresetthe{remark}

\numberwithin{equation}{section}
\numberwithin{figure}{section}
\numberwithin{table}{section}

\crefname{section}{section}{sections}
\Crefname{section}{Section}{Sections}
\crefname{assumption}{assumption}{assumptions}
\Crefname{assumption}{Assumption}{Assumptions}

\newcommand{\E}{\mathbb E}
\newcommand{\Prb}{\mathbb P}
\newcommand{\R}{\mathbb R}

\newcommand{\one}{\mathbf 1}

\DeclareMathOperator*{\argmax}{arg\,max}
\DeclareMathOperator*{\argmin}{arg\,min}
\newcommand{\cE}{\mathcal E}
\newcommand{\CA}{C^{A}}
\newcommand{\CR}{C^{R}}
\newcommand{\CRC}{C^{RC}}

\newcommand{\I}{\mathcal I}

\newcommand{\cM}{\mathcal M}

\newcommand{\cZ}{\mathcal Z}

\newcommand{\Safe}{\mathrm{Safe}}
\newcommand{\Risky}{\mathrm{Risky}}
\newcommand{\doi}[1]{\href{https://doi.org/#1}{\nolinkurl{doi:#1}}}

\newenvironment{sectionoverview}
  {\begin{quote}\small\noindent\textbf{Section overview.}\ }
  {\end{quote}\vspace{0.3em}}

\setlist{nosep,leftmargin=1.8em}
\titleformat{\section}
  {\normalfont\Large\bfseries}
  {\thesection}{0.7em}{}

\titleformat{\subsection}
  {\normalfont\large\bfseries}
  {\thesubsection}{0.7em}{}

\titleformat{\subsubsection}
  {\normalfont\normalsize\bfseries}
  {\thesubsubsection}{0.7em}{}

\titlespacing*{\section}{0pt}{2.8ex plus .7ex}{1.1ex}
\titlespacing*{\subsection}{0pt}{2.0ex plus .5ex}{0.8ex}
\titlespacing*{\subsubsection}{0pt}{1.5ex plus .3ex}{0.6ex}

\fancypagestyle{plain}{
  \fancyhf{}
  \fancyfoot[C]{\thepage}
  
}

\title{New Approaches to Strategic Thinking in Dynamic Games\thanks{Department of Political Economy, King's College London, London, UK. mehmet.mars.seven@kcl.ac.uk.}}
\author{Mehmet Mars Seven}
\date{\today}

\begin{document}

\maketitle

\begin{abstract}
We present several models of strategic choice concerning equilibrium reasoning and behaviour, as well as new applications to dynamic games. We propose concepts such as historical equilibrium, simple Nash equilibrium, cautious backward induction, and an approach inspired by how chess players reason in sequential games. \textit{JEL codes:} C72, C73, D03.
\end{abstract}

\noindent\textbf{Keywords:} extensive-form games, solution concepts, imperfect recall, chess, maximin

\tableofcontents
\clearpage

\section{Thinking While the Opponent Thinks}
\label{sec:time}

\subsection{Introduction}\label{time:sec:intro}

In a clocked alternating-move game, deliberation can improve a player's preparation while simultaneously preparing the opponent. This section models that tradeoff in a finite extensive-form game where both players prepare during the mover's thinking time, but only the mover's clock falls. It characterizes a temporary exploitation window in which an attack becomes available before its defense and, when the opponent's reply set expands with deliberation, gives conditions under which the mover can verify the best immediate move before the deadline. The analysis distinguishes equilibrium play from completed deliberation: a move may be optimal even when its optimality has not yet been established by the prescribed analysis.

Concurrent deliberation creates a preparation externality: the mover pays for time that both players can use. In the public-state model, standard backward induction characterizes equilibrium over admissible physical moves and deliberation. Under the stated threshold restrictions, that equilibrium can require movement inside a temporary exploitation window.

The model is related to the literature on bounded reasoning and strategic deliberation \citep{Jehiel1995,LarsonSandholm2001,HalpernPass2015,Shannon1950,Himstedt2005,Orton2021}. Our focus is the timing rule that gives both players preparation while charging only the mover. We take deliberation rules and admissible moves as primitives. Evidence that chess players devote more time to positions where further computation is more valuable motivates endogenous timing \citep{RussekEtAl2025}.

\subsection{The game}\label{time:sec:game}

Let $\Gamma=(H,Z,I(\cdot),(A(h))_{h\notin Z},u)$ be a finite two-player perfect-information game without chance. Here $H$ is the set of physical histories, $Z\subseteq H$ the terminal histories, $I(h)\in I=\{1,2\}$ the player on move, $A(h)$ the finite set of legal physical actions, and $u_i:Z\to\mathbb R$ player $i$'s terminal payoff. Write $ha$ for the physical history after action $a$ at $h$.

Player $i$ has a finite preparation-state space $K_i$ and a clock $c_i\in\{0,\ldots,T_i\}$, where $T_i$ is a nonnegative integer. A state $k_i\in K_i$ records retained analysis and the preparation on which admissible moves depend. Initial preparation states are specified. Three commonly known maps describe the technology:
\begin{enumerate}[label=(\roman*)]
\item $D_i:H\times K_i\to K_i$ is the state update after one unit of deliberation;
\item $\varnothing\ne\cM_i(h,k_i)\subseteq A(h)$ is the set of admissible physical moves when $I(h)=i$;
\item
\[
R_i:\{(k_i,h,a):k_i\in K_i,\ h\in H\setminus Z,\ a\in A(h)\}\to K_i
\]
is the preparation retained after physical action $a$.
\end{enumerate}
The correspondence $\cM_i$ is a restriction on feasible play. An action outside it cannot be selected in the expanded game. The maps may embody a fixed order of analysis or an order that depends on earlier evaluations recorded in $k_i$. Because $D_i$ depends on the physical history $h$, the same update map may advance preparation at different rates depending on whose physical turn it is. Thus the own-turn and background rates used below are special cases of $D_i$. The nonmover's update is automatic.

The public state is $s=(h,c_1,c_2,k_1,k_2)$. Both players observe this full state after every transition. At $i=I(h)$, player $i$ may choose $a\in\cM_i(h,k_i)$, leading to
\begin{equation}\label{time:eq:move}
\mu(s,a)=\bigl(ha,c_1,c_2,R_1(k_1,h,a),R_2(k_2,h,a)\bigr).
\end{equation}
If $c_i>0$, she may instead deliberate for one unit. For player 1,
\begin{equation}\label{time:eq:think1}
\Theta_1(s)=\bigl(h,c_1-1,c_2,D_1(h,k_1),D_2(h,k_2)\bigr),
\end{equation}
and for player 2,
\begin{equation}\label{time:eq:think2}
\Theta_2(s)=\bigl(h,c_1,c_2-1,D_1(h,k_1),D_2(h,k_2)\bigr).
\end{equation}
Physical moves take no clock time. A player with zero clock must select an admissible physical move, but may still prepare during the opponent's deliberation. Payoffs depend only on the terminal physical history. There is no direct cost of thinking beyond the clock constraint and its effect on preparation. Unless stated otherwise, the applications below begin at an initial state with $c_i=T_i$.

A strategy specifies Move or Deliberate, and the physical action when applicable, after every public history of this expanded game. All feasible history-contingent strategies are allowed. The solution concept is subgame perfection. No private preparation states, private computational discoveries, or uncertainty about payoffs are added to this equilibrium model.

\begin{proposition}[Finite representation]\label{time:prop:existence}
The expanded game has a pure subgame-perfect equilibrium. In the zero-sum case, let $V(s)$ be player 1's continuation value. At terminal $h$, $V(s)=u_1(h)$. At player-1 states,
\begin{equation}\label{time:eq:recursion1}
V(s)=\max\left\{\max_{a\in\cM_1(h,k_1)}V(\mu(s,a)),\ V(\Theta_1(s))\right\},
\end{equation}
where the deliberation term is omitted at $c_1=0$. At player-2 states,
\begin{equation}\label{time:eq:recursion2}
V(s)=\min\left\{\min_{a\in\cM_2(h,k_2)}V(\mu(s,a)),\ V(\Theta_2(s))\right\},
\end{equation}
where the deliberation term is omitted at $c_2=0$.
\end{proposition}

\begin{proof}
Let $d(h)$ be the maximum number of physical moves remaining after $h$. Each action strictly reduces $d(h)+c_1+c_2$: deliberation reduces one clock, and a physical move reduces $d(h)$. Thus the expanded tree is finite. All decisions and states are public, so backward induction yields the asserted equilibrium and recursion.
\end{proof}

This is an application of finite-game backward induction. The public-state assumption is substantive: if preparation states were privately observed, the expanded game would generally have imperfect information and Proposition~\ref{time:prop:existence} would no longer follow from ordinary backward induction.

Preparation matters here through the admissibility restrictions. Indeed, in the zero-sum case, if $\cM_i(h,k_i)=A(h)$ everywhere, then $V(s)$ equals the value of the underlying physical game, and moving immediately according to physical-game backward induction is optimal. With state-dependent admissibility, the same physical position and clocks can instead have different values at different preparation states.

\subsection{An example}\label{time:sec:illustration}

The distinction between a solution and a completed calculation is visible in a single game. In Figure~\ref{time:fig:refutation-game}, player 1 chooses Safe or Risky. Safe gives $(0,0)$. After Risky, player 2 chooses among $q\ge3$ replies. Replies $b_1,\ldots,b_{q-1}$ give $(1,-1)$, while $b_q$ gives $(-1,1)$. Backward induction in the physical game prescribes $b_q$ and therefore Safe.

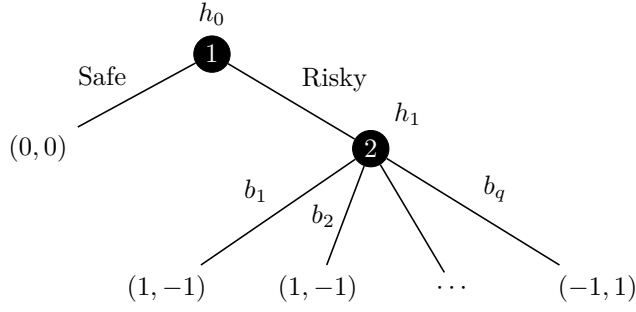
\begin{figure}[htbp]
\centering
\begin{tikzpicture}[
    x=1cm,y=1cm,
    every node/.style={font=\small},
    player/.style={circle,draw,fill=black,text=white,inner sep=2pt},
    edge/.style={line width=.6pt}
]
\node[player,label=above:{$h_0$}] (h0) at (0,0) {1};
\node (z0) at (-2.3,-1.25) {$(0,0)$};
\node[player,label=above right:{$h_1$}] (h1) at (2.1,-1.25) {2};
\node (z1) at (-0.6,-3.1) {$(1,-1)$};
\node (z2) at (1.4,-3.1) {$(1,-1)$};
\node (dots) at (3.2,-3.1) {$\cdots$};
\node (zq) at (5.1,-3.1) {$(-1,1)$};
\draw[edge] (h0)--node[above left]{$\Safe$}(z0);
\draw[edge] (h0)--node[above right]{$\Risky$}(h1);
\draw[edge] (h1)--node[above left]{$b_1$}(z1);
\draw[edge] (h1)--node[left]{$b_2$}(z2);
\draw[edge] (h1)--(dots);
\draw[edge] (h1)--node[above right]{$b_q$}(zq);
\end{tikzpicture}
\caption{A fixed two-stage zero-sum game. The refutation is $b_q$, and Safe is the physical-game backward-induction action.}
\label{time:fig:refutation-game}
\end{figure}

Suppose player 1 starts without a completed evaluation of this continuation. Her deliberation rule examines $b_1,b_2,\ldots,b_q$ in that order, one reply per clock unit, and records the resulting comparisons in $k_1$. No previously established comparison supplies the value of an unexamined reply. If her clock permits only $m<q$ units, this calculation does not reach $b_q$ before she must move. Its record therefore does not establish the negative continuation value of Risky.

Note that this is not a claim that Safe is impossible to choose, or that every way of analyzing this particular game takes $q$ steps. Safe may be admitted from the outset and chosen without a completed calculation. The public equilibrium analysis in Section~\ref{time:sec:game} also has access to the full primitives. The distinction is between the action prescribed by the analyst and the comparisons recorded by the stipulated deliberation rule.

In what follows, a move is \emph{verified by deliberation} when the completed evaluations and retained comparisons establish that it is a best response to the specified continuation. Verification is a requirement on the analysis. A known answer retained before play counts as preparation. Absent such preparation, a clock can force a legal move before the prescribed analysis has established its optimality.

\subsection{Strategic thinking and the exploitation window}\label{time:sec:window}

Now let preparation affect feasible play. Player 1 chooses Safe, worth zero, or an attack $a$, after which player 2 replies and the game ends. Safe is always admissible. The attack becomes admissible when player 1 has accumulated $d_a>0$ units of relevant analysis. A defending reply becomes admissible to player 2 at $c_a>0$ units of branch-specific preparation. Before then, every admissible reply gives player 1 the same payoff $w>0$; once a defense is admissible, player 2 can obtain a player-1 payoff at most zero. Admissible reply sets expand with retained preparation.

The specialization starts with the full clocks $c_1=T_1$ and $c_2=T_2$. Player 1 has initial preparation $M_{1,a}\ge0$ and rate $r_1>0$. Player 2 has initial preparation $M_{2,a}\ge0$, rate $\lambda_a\ge0$ on this branch during player 1's deliberation, and rate $r_2>0$ on her own turn. The quantities $d_a$, $M_{1,a}$, and $r_1t$ are measured in a common player-1 analysis unit; the quantities $c_a$, $M_{2,a}$, $\lambda_at$, and $r_2T_2$ are measured in a common player-2 branch-preparation unit. The two players' units need not be comparable to each other. A physical move retains the branch-specific work. Thus, if player 1 moves at integer time $t$, player 2 can have $M_{2,a}+\lambda_a t+r_2T_2$ units before replying. These rates are a reduced-form specialization of the history-dependent maps $D_i$: player 2's update during a player-1 history may add $\lambda_a$, whereas her update during her own physical turn may add $r_2$. These are restrictions on $D_i$, $R_i$, and $\cM_i$.

For a real number $x$, write $[x]_+=\max\{x,0\}$. Define the first integer dates at which the attack and its eventual defense become feasible:
\begin{align}
 t_a^A&=\left\lceil\frac{[d_a-M_{1,a}]_+}{r_1}\right\rceil,\label{time:eq:attacktime}\\
 t_a^D&=\inf\{t\in\mathbb Z_{\ge0}:M_{2,a}+\lambda_a t+r_2T_2\ge c_a\}.\label{time:eq:defensetime}
\end{align}
The infimum of the empty set is $+\infty$. If $\lambda_a>0$, then
$t_a^D=\lceil[c_a-M_{2,a}-r_2T_2]_+/\lambda_a\rceil$.
If $\lambda_a=0$, the defense date is zero when the initial and own-turn work suffice, and $+\infty$ otherwise.

\begin{proposition}[Exploitation window]\label{time:prop:window}
In this two-stage game, player 1 can obtain $w$ against an optimal continuation by player 2 if and only if
\begin{equation}\label{time:eq:window}
t_a^A\le T_1\quad\text{and}\quad t_a^A<t_a^D.
\end{equation}
The profitable integer movement dates are exactly $t_a^A\le t\le T_1$ with $t<t_a^D$. A tie-breaking convention favoring earlier movement among equal terminal payoffs selects $t_a^A$ uniquely; without it, every profitable date is optimal.
\end{proposition}

\begin{proof}
The attack is admissible exactly when $t\ge t_a^A$. Because replies expand with preparation and thinking has no direct cost, player 2 can attain her best available reply by using enough of her own clock, up to $T_2$. A defense is attainable by that deadline exactly when $t\ge t_a^D$. Hence an admissible attack gives $w$ precisely at the stated dates. At all other dates it is unavailable or is weakly dominated by Safe in continuation payoff. The earliest profitable date is $t_a^A$.
\end{proof}

Greater preparation or a faster rate for player 1 weakly advances $t_a^A$. Greater preparation, remaining time, or branch attention for player 2 weakly advances $t_a^D$ and contracts the window. The qualifiers are weak because time is discrete. The rate $\lambda_a$ is part of the specified deliberation technology; its strategic allocation across several attacks is not determined by Proposition~\ref{time:prop:window}.

\begin{example}[An interior equilibrium movement date]\label{time:ex:tactical}
Let $T_1=4$, $T_2=1$, and $w=1$. Suppose player 1 has no initial preparation on the attack, $M_{1,a}=0$, needs $d_a=4$ units before the attack is admissible, and accumulates analysis at rate $r_1=2$. Then
\[
t_a^A=\left\lceil\frac{4}{2}\right\rceil=2.
\]
Suppose player 2 has no initial branch preparation, $M_{2,a}=0$, needs $c_a=4$ units before a defense is admissible, accumulates one unit during each period of player 1's deliberation, $\lambda_a=1$, and one unit on her own turn, $r_2=1$. Then
\[
t_a^D=\inf\{t\in\mathbb Z_{\ge0}:t+1\ge4\}=3.
\]
Thus the exploitation window consists of the single date $t=2$. At $t=0$ and $t=1$ the attack is not yet admissible, so moving immediately gives Safe and payoff zero. At $t=2$ the attack is admissible while the defense is still unavailable by player 2's deadline, so moving yields $1$. At $t=3$ or $t=4$ the defense is available and Safe is weakly better. Hence, starting from $t=0$, subgame-perfect play has player 1 deliberate twice and then use the attack at $t=2$. This is a literal numerical specialization of Proposition~\ref{time:prop:window}; the expanding-reply-set issue is treated separately in Section~\ref{time:sec:verification}.
\end{example}

\subsection{Verification with an expanding response set}\label{time:sec:verification}

The exploitation window concerns payoffs in the expanded game. We now ask a separate question: can the mover's completed analysis establish the best response to the opponent's eventual reply set? The target here is the optimal immediate physical move at the chosen date.

Consider a two-stage continuation that starts with the full clocks $c_1=T_1$ and $c_2=T_2$, with Safe worth zero and $N\ge1$ ordered replies $b_1,\ldots,b_N$ after Risky. The order is common to the two deliberation rules: the scalar preparation measures below count progress through these same nested replies. If player 1 moves after $t$ units of deliberation, player 2 can select a best reply within
\begin{equation}\label{time:eq:responses}
B(t)=\{b_1,\ldots,b_{q(t)}\},\qquad
q(t)=\min\{N,M_2+\lambda t+r_2T_2\}.
\end{equation}
Here $M_2$ is retained preparation, $\lambda$ is the relevant background rate, and $r_2$ is the own-turn rate. Assume $M_2,r_2,\lambda$ are nonnegative integers, $r_2>0$, and $1\le q(t)\le N$ at every integer $t\in\{0,\ldots,T_1\}$. The response rule is fixed: player 2 uses a best reply in $B(t)$, and no reply outside $B(t)$ can enter the response, whether as an evaluated move or a fallback. Such a rule is consistent with nested admissibility sets and terminal own-turn deliberation. Any initially admissible fallback must already belong to $B(t)$.

Let $v_j=u_1(h_1b_j)\in\{-1,1\}$ be the payoff after reply $b_j$, where $h_1$ is the history following Risky. A reply with $v_j=-1$ is a refutation. The unique best immediate root action against the specified response rule is
\begin{equation}\label{time:eq:target}
a^*(t)=
\begin{cases}
\Risky,&\min_{1\le j\le q(t)}v_j=1,\\
\Safe,&\min_{1\le j\le q(t)}v_j=-1.
\end{cases}
\end{equation}
Both root actions are admitted for this comparison. When $q(t)<N$, \eqref{time:eq:target} is the best immediate action against the bounded response set $B(t)$; it need not coincide with the backward-induction action in the unrestricted physical game. Only when the relevant response set coincides with the full continuation does the distinction disappear.

Player 1 retains completed evaluations of the first $M_1$ replies, with $0\le M_1\le N$, and can evaluate $r_1$ additional replies per period, where $r_1$ is a positive integer. For this verification exercise, any prescribed order is feasible, completed evaluations are retained, and the resulting payoff comparisons require no further clock charge. An order may depend on earlier evaluations; once specified, that dependence is part of the deliberation rule. Thus the first $\min\{N,M_1+r_1t\}$ replies can be completely evaluated by time $t$. The number $M_1$ denotes these particular retained evaluations.

Verification has its meaning from Section~\ref{time:sec:illustration}: the completed analysis must establish the comparison in \eqref{time:eq:target}. In this setting, finding one relevant refutation establishes Safe; evaluating every reply in $B(t)$ and finding no refutation establishes Risky. A verification must remain valid for every assignment of unresolved continuation payoffs consistent with the completed comparisons. This is a requirement on the calculation. The following assumption supplies the necessary restriction on other ways of completing the comparison.

\begin{assumption}[Common unresolved continuation]\label{time:ass:common}
There is an admissible continuation in which every reply after Risky gives player 1 payoff $1$. Along the course of deliberation in that continuation, at every date and for every permitted order of examination, the root comparison remains unverified whenever some reply in $B(t)$ has not been evaluated. No retained conclusion, indirect deduction, or observation of the opponent supplies the missing comparison beyond the evaluations counted above. The same continuation supports this course of deliberation through the entire deadline.
\end{assumption}

The assumption is a restriction on what the specified deliberation can establish. It holds, for example, when every unexamined reply is left unresolved as a possible sole refutation by the calculation, and the record changes only through evaluating replies. On the course where all evaluations favor Risky, any unexamined relevant reply then prevents verification. Other courses may terminate much earlier when a refutation is found. This does not introduce private payoff information into the public equilibrium model: it specifies which comparisons must actually have been completed for the separate verification requirement.

A stopping rule may condition its movement date on all evaluations already completed, but must move no later than $T_1$. A guarantee of verification requires its current comparison to be established whenever it stops, for every continuation permitted by this requirement. The common continuation in Assumption~\ref{time:ass:common}, rather than a different difficult continuation at each date, is what makes a bound on such stopping rules possible.

\begin{theorem}[Verification frontier]\label{time:thm:frontier}
Under the preceding evaluation and response restrictions and Assumption~\ref{time:ass:common}, a stopping rule guaranteeing verification of \eqref{time:eq:target} by the deadline exists if and only if
\begin{equation}\label{time:eq:some-time}
M_1+r_1t\ge q(t)\quad\text{for some }t\in\{0,\ldots,T_1\}.
\end{equation}
Equivalently,
\begin{equation}\label{time:eq:endpoints}
\begin{split}
M_1&\ge\min\{N,M_2+r_2T_2\},\\
&\hspace{1em}\text{or}\\
M_1+r_1T_1&\ge\min\{N,M_2+\lambda T_1+r_2T_2\}.
\end{split}
\end{equation}
If both inequalities fail, no stopping rule based on earlier evaluations can guarantee verification. Randomizing the order of examination or the stopping date cannot restore a probability-one guarantee subject to the same deadline and evaluation limits.
\end{theorem}

\begin{proof}
For sufficiency, choose a date satisfying \eqref{time:eq:some-time}. The retained initial segment and the available additional evaluations suffice to examine all of $B(t)$ by that date. Comparing the resulting minimum with zero establishes \eqref{time:eq:target}. A preassigned stopping date therefore suffices.

For necessity, suppose $M_1+r_1t<q(t)$ at every feasible date. Take any stopping rule and follow its deliberation in the common continuation specified by Assumption~\ref{time:ass:common}. Let $\tau\le T_1$ be its stopping date on that course. At most $M_1+r_1\tau$ distinct replies have been evaluated. Hence some reply in $B(\tau)$ remains unevaluated. By the assumption, the comparison is still unverified at $\tau$. This argument follows one continuation until the rule itself chooses to stop; it does not select the continuation after fixing a stopping date. For a randomized rule, the same conclusion holds for every realized order and stopping date on that common continuation, so verification cannot hold with probability one there.

For the endpoint characterization, extend
\[
F(t)=M_1+r_1t-\min\{N,M_2+\lambda t+r_2T_2\}
\]
to real $t\in[0,T_1]$. Equivalently,
\[
F(t)=\max\{M_1+r_1t-N,\ M_1-M_2-r_2T_2+(r_1-\lambda)t\}.
\]
It is the maximum of two affine functions and is therefore convex. Its maximum on the interval is attained at an endpoint. Both endpoints are permitted integer dates, and the formula agrees with the integer evaluation counts. Thus \eqref{time:eq:some-time} is equivalent to $F(0)\ge0$ or $F(T_1)\ge0$, which gives \eqref{time:eq:endpoints}.
\end{proof}

The common unresolved-continuation assumption is essential to the necessity argument. A statement that some continuation requires $q(t)$ evaluations at each preassigned date is not enough: a stopping rule may stop early when the early comparison is easy and continue when a later comparison is easy. Nor does it suffice that the set of payoffs favoring Safe expands with the response set. Theorem~\ref{time:thm:frontier} instead uses one unresolved course at every possible stopping date. It does not assert that every game or every course of deliberation needs all $q(t)$ evaluations.

The preparation and integer-time restrictions matter as well. Retained evaluations outside $B(t)$ need not help establish \eqref{time:eq:target}; a previously established solution may help much more than its length suggests. The scalar credit $M_1$ is exact here because it consists of the initial segment of evaluated replies. If the rates require rounding, the feasible-date condition must be checked directly with the actual integer evaluation counts: rounding need not preserve the endpoint argument.

\begin{corollary}[Equal rates without preparation]\label{time:cor:equal}
Under the assumptions of Theorem~\ref{time:thm:frontier}, suppose $M_1=M_2=0$, $r_1=r_2=\lambda=r>0$, and $T_2>0$. A guarantee of verification by player 1's deadline is possible if and only if
\begin{equation}\label{time:eq:equal}
N\le rT_1.
\end{equation}
\end{corollary}

\begin{proof}
Before the response set reaches all $N$ replies, it has size $rt+rT_2$, whereas player 1 can have evaluated only $rt$ replies. The gap is $rT_2>0$. The first endpoint condition fails, and the second holds exactly when $rT_1\ge N$. Apply Theorem~\ref{time:thm:frontier}.
\end{proof}

The result concerns guaranteed verification of the immediate best reply to the stipulated bounded response set. An early refutation, a retained solution, or an additional valid inference can permit earlier verification when Assumption~\ref{time:ass:common} does not apply. Even where the assumption applies, player 1 can still choose an admissible move without verification. Neither the frontier nor the corollary rules out equilibrium or profitable play.

\paragraph{The terminal second mover.}
After a branch $a$ is chosen, the terminal responder faces a fixed comparison. Suppose $C_{2,a}$ units of retained branch-specific analysis are necessary and sufficient for her deliberation rule to establish an exact best reply. With initial preparation $M_{2,a}$, background rate $\lambda_a$, and own-turn rate $r_2>0$, the additional number of own-clock periods required after movement at $t$ is
\[
\left\lceil\frac{[C_{2,a}-M_{2,a}-\lambda_a t]_+}{r_2}\right\rceil.
\]
Completion is possible precisely when this number is at most $T_2$. With retained analysis, expanding admissible replies, and no direct thinking cost, continuing until the deadline cannot worsen her best available reply. At nonterminal positions, further delay may again prepare an opponent who will move later.

\clearpage
\section{Historical Equilibrium: Rational Forward Inference and Backward Induction}

\begin{sectionoverview}
    This section introduces Historical Equilibrium (HE), an extensive-form solution concept that combines backward evaluation of continuation payoffs with forward inference from observed play. Players choose complete behavioral plans together with interaction weights that generate a joint law over contingent actions. HE is Nash equilibrium in this game. The section provides applications to increasing-sum and constant-sum centipede games.
\end{sectionoverview}

\subsection{Introduction}

Historical Equilibrium (HE) is an extensive-form solution concept in which players choose complete behavioral plans together with interaction weights. These inputs generate a joint law over complete contingent behavior. An observed history is interpreted by inserting its realized actions into the same interaction system and renormalizing, so every feasible history---including a zero-probability surprise---has a coherent continuation law. HE is Nash equilibrium in this game. It combines backward evaluation of continuation payoffs with \emph{forward inference}: observed actions reveal information about complete plans, and players anticipate that rational inference when choosing their plans. This section gives the definition, a finite mixed-existence result, and applications to centipede games \citep{Rosenthal1981}.

Backward induction treats a reached history as a separate continuation problem, with the preceding path strategically irrelevant. HE uses a different unilateral counterfactual: a player changes an entire behavioral-and-interaction plan while opponents' inputs remain fixed, and the induced joint law is recomputed. Consequently, a player’s planned action at a later node can be associated with an earlier response, so changing that planned action can also affect the probability that the later node is reached.

An observed continuation eliminates complete plans that would have stopped, thereby changing the posterior distribution over the plans still compatible with play. This is rational inference. It becomes forward inference when the observed action is used to infer future behavior from the underlying complete plan. Players choose their plans anticipating that opponents will make exactly this inference. At a fixed plan profile, continuation payoffs may still be evaluated backward; what HE rejects is a fresh local maximization that discards the informational content of the history. Thus HE combines backward evaluation with forward inference.

\subsection{Historical Equilibrium}

Consider a finite extensive-form game. Let $I$ be its player set, $D$ its decision nodes, $A(h)$ the finite action set at node $h$, and $I(h)$ the active player. A complete contingent action configuration is $\omega\in\Omega:=\prod_{h\in D}A(h)$; it induces a terminal history $z(\omega)$ and material payoff $u_i(z(\omega))$.

Player $i$ chooses a complete HE plan $\tau_i=(\sigma_i,\psi_i)$. Its behavioral part consists of \emph{baseline probabilities} $\sigma_h\in\Delta(A(h))$ at every node controlled by $i$. Its interaction part consists of positive interaction weights $\psi_{iF}:A_F\to[\underline\psi_{iF},\overline\psi_{iF}]$, where each scope $F\subseteq D$ is an admissible set of path-compatible decision nodes whose actions may interact. Let $\mathcal T_i$ denote the set of all such feasible pairs $(\sigma_i,\psi_i)$. Because the game is finite and each interaction weight takes values in a closed bounded interval, $\mathcal T_i$ is compact. Given $\tau=(\tau_i)_{i\in I}$, define
\begin{equation}\label{he:eq:root-law}
 W_\tau(\omega)=\prod_{h\in D}\sigma_h(\omega_h)\prod_{F}\prod_i\psi_{iF}(\omega_F),
 \qquad
 Q_\tau^{\varnothing}(\omega)=\frac{W_\tau(\omega)}{\sum_{\eta\in\Omega}W_\tau(\eta)}. 
\end{equation}
The first product contains baseline probabilities; the second contains interaction weights. No mediator selects $Q_\tau^{\varnothing}$: it is generated by the players' separate choices.

For a feasible history $\rho$, let $D(\rho)$ be the nodes whose actions are fixed by $\rho$. For a remaining action assignment $\omega_{-D(\rho)}$, insert the observed actions into every interaction weight, omit the already realized baseline factors, and set
\begin{equation}\label{he:eq:history-law}
 \widetilde W_\tau^{\rho}(\omega_{-D(\rho)})=
 \prod_{g\notin D(\rho)}\sigma_g(\omega_g)
 \prod_F\prod_i\psi_{iF}\bigl(\rho_{F\cap D(\rho)},\omega_{F\setminus D(\rho)}\bigr),
 \quad
 Q_\tau^{\rho}=\frac{\widetilde W_\tau^{\rho}}{\sum\widetilde W_\tau^{\rho}}. 
\end{equation}
Strict positivity of the interaction weights makes the denominator positive. If $\rho$ has positive root probability, \eqref{he:eq:history-law} is ordinary Bayesian conditioning. If it has zero root probability, \eqref{he:eq:history-law} is the structural continuation of the same model. Hence surprises can be unlikely, but they are never uninterpreted.

\begin{definition}[Historical Equilibrium]\label{he:definition:auto-1}
A plan profile $\tau^*\in\prod_i\mathcal T_i$ is a \emph{Historical Equilibrium} if, for every player $i$,
\begin{equation}\label{he:eq:he}
 \tau_i^*\in\arg\max_{\tau_i\in\mathcal T_i}
 \E_{Q_{(\tau_i,\tau_{-i}^*)}^{\varnothing}}
 \bigl[u_i(z(\omega))\bigr]. 
\end{equation}
The HE includes the complete family of structural kernels $\{Q_{\tau^*}^{\rho}\}_{\rho}$. A deviation holds opponents' baseline probabilities and interaction weights fixed, but generally changes the normalized law and therefore their induced conditional behavior.
\end{definition}

\begin{theorem}[Finite mixed existence]\label{he:theorem:auto-1}
Every finite extensive-form game with compact HE plan sets and continuous strictly positive interaction-weight functions has a mixed Historical Equilibrium.
\end{theorem}

\begin{proof}
The product $\prod_i\mathcal T_i$ is a product of compact metric spaces. Because $\Omega$ is finite and the denominator in \eqref{he:eq:root-law} is strictly positive, each root expected payoff is continuous in the complete plan profile. Glicksberg's (\citeyear{Glicksberg1952}) extension of Nash's (\citeyear{nash1950}) theorem therefore gives a Nash equilibrium in Borel probability measures over the compact plan sets. By \eqref{he:eq:he}, this is a mixed HE.
\end{proof}

\subsection{The standard pairwise technology}

All numerical examples use one standard positive binary interaction: the log-linear, or Gibbs, odds-ratio parameterization. Suppose two successive decisions are $S$ or $C$. Let $y$ and $x$ be the first and second mover's baseline probabilities of $C$. Each player chooses a bounded interaction factor $\kappa_i\in[1,2]$, and $K:=\kappa_1\kappa_2\in[1,4]$. The unnormalized law is
\begin{equation}\label{he:eq:pairwise}
\begin{array}{c|cc}
 &S_2&C_2\\ \hline
S_1&(1-y)(1-x)&(1-y)x\\
C_1&y(1-x)&Kyx
\end{array}
\qquad\text{with}\qquad \mathcal Z=1+(K-1)xy. 
\end{equation}
After division by $\mathcal Z$, the odds ratio is exactly $K$; equivalently, the interaction term is $\log K$ in a binary log-linear model. Thus $K=1$ is independence and $K>1$ represents represents positive interaction between continuation actions. Other coherent positive interaction-weight or response-kernel technologies could be used, but \eqref{he:eq:pairwise} is the only technology used in this section's examples.

\subsection{The three-node centipede}

Consider

\begin{figure}[ht]
\centering
\begin{tikzpicture}[
  node distance=18mm and 25mm,
  decision/.style={circle,draw,inner sep=1.7pt},
  terminal/.style={rectangle,draw=none,inner sep=1pt}
]
\node[decision] (p1a) {$P_1$};
\node[terminal,below left=of p1a] (z1) {$(2,0)$};
\node[decision,below right=of p1a] (p2) {$P_2$};
\node[terminal,below left=of p2] (z2) {$(0,3)$};
\node[decision,below right=of p2] (p1b) {$P_1$};
\node[terminal,below left=of p1b] (z3) {$(4,1)$};
\node[terminal,below right=of p1b] (z4) {$(3,5)$};
\draw[-{Latex[length=2mm]}] (p1a) -- node[left] {$S_1$} (z1);
\draw[-{Latex[length=2mm]}] (p1a) -- node[right] {$C_1$} (p2);
\draw[-{Latex[length=2mm]}] (p2) -- node[left] {$S_2$} (z2);
\draw[-{Latex[length=2mm]}] (p2) -- node[right] {$C_2$} (p1b);
\draw[-{Latex[length=2mm]}] (p1b) -- node[left] {$S_3$} (z3);
\draw[-{Latex[length=2mm]}] (p1b) -- node[right] {$C_3$} (z4);
\end{tikzpicture}
\caption{The three-decision centipede.}
\label{he:eq:centipede-tree}
\end{figure}
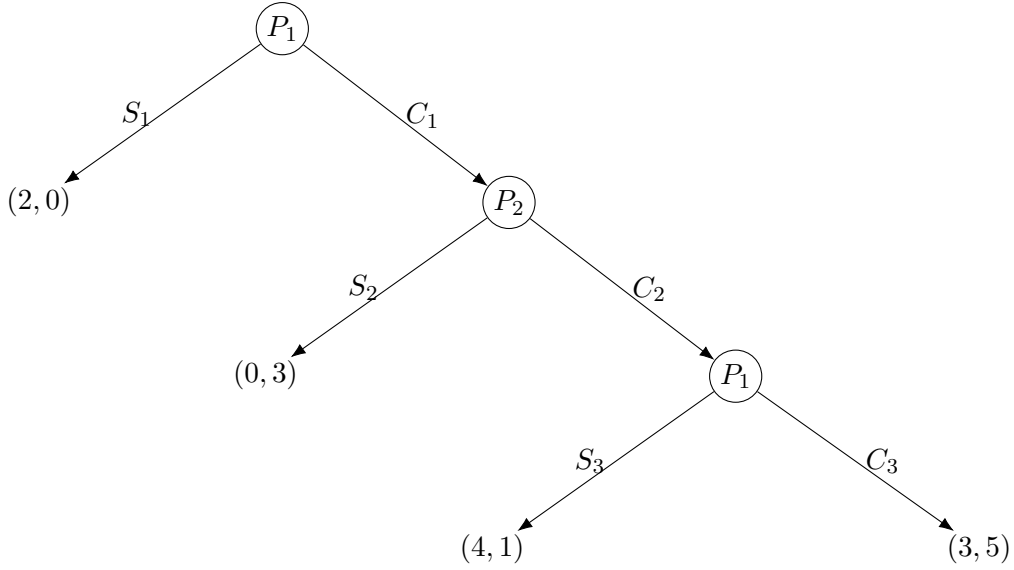

Let $t$ be player 1's baseline probability of $C_1$, $y$ player 2's baseline probability of $C_2$, and $x$ player 1's baseline probability of $C_3$. Conditional on $C_1$, apply \eqref{he:eq:pairwise} to $(C_2,C_3)$. The continuation payoffs are
\begin{equation}\label{he:eq:continuation}
 U_1^{C_1}=\frac{4y(1-x)+3Kxy}{1+(K-1)xy},
 \qquad
 U_2^{C_1}=\frac{3(1-y)+y(1-x)+5Kxy}{1+(K-1)xy}. 
\end{equation}
For fixed $K$, player 2's interior best-response condition in $y$ gives $x(K)=1/(K+1)$, while player 1's interior condition in $x$ gives $y(K)=(3K-4)/[4(K-1)]$. At that candidate, both players strictly prefer a larger interaction factor, so $\kappa_1^*=\kappa_2^*=2$ and $K^*=4$. Player 1's continuation value is then $8/3>2$, so $t^*=1$.

\begin{proposition}[Three-node solution]\label{he:proposition:auto-1}
The pairwise game has an immediate-stopping HE class and one positive-continuation HE. In the latter, the auxiliary choices are $t^*=1$, $x^*=1/5$, $y^*=2/3$, and $\kappa_1^*=\kappa_2^*=2$. Conditional on $C_1$, the normalized law is
\begin{equation}\label{he:eq:centipede-law}
 Q^*(\cdot\mid C_1)=\frac1{21}
 \begin{array}{c|cc}
 &S_3&C_3\\ \hline
 S_2&4&1\\
 C_2&8&8
 \end{array}. 
\end{equation}
Hence $\Prb(C_2\mid C_1)=16/21$, $\Prb(C_3\mid C_1C_2)=1/2$, and expected payoffs are $(8/3,3)$. This positive HE strictly Pareto-dominates the stopping class.
\end{proposition}

The distinction between baseline and realized behavior is essential: $x^*=1/5$ is not the probability of continuation once the final node is reached; rational inference from the history raises that conditional probability to $1/2$. Observing $C_1$ first filters out player 1's complete plan that stops immediately. Observing $C_2$ then shifts the posterior toward the plan containing $C_3$. Player 1 chooses a complete plan while anticipating this forward inference by player 2.

\subsection{When HE stops}

\subsection*{A large-stakes positive-sum example}

HE does not mechanically select continuation. Consider the same tree with
\begin{equation}\label{he:eq:large-stakes}
 S_1\mapsto(100,0),\quad C_1S_2\mapsto(0,100),\quad
 C_1C_2S_3\mapsto(200,10),\quad C_1C_2C_3\mapsto(40,250). 
\end{equation}
For the continuation block, player 1's final stopping and continuing payoffs are $L=200$ and $c=40$. The interior condition is $y=(Kc-L)/[(K-1)L]$, so a positive pairwise HE requires $Kc>L$. Under the maintained bound $K\le4$, however, $Kc\le160<200$. Therefore the continuation block has no positive HE: player 2 stops after $C_1$, and player 1 strictly prefers the root payoff $100$ to entering. The only HE outcome is the all-stop outcome, $(S_1,S_2,S_3)$. Inactive off-path baseline probabilities and interaction weights can be payoff-irrelevant, so the uniqueness claim is deliberately outcome-wise. Notice that full continuation has the largest total payoff, $290$, yet the interaction technology is not strong enough to overcome player 1's large final stopping temptation.

\subsection*{Constant-sum centipedes}

A two-player centipede is constant sum if every terminal payoff satisfies $u_1(z)+u_2(z)=c$. A technology is \emph{baseline-support preserving} if an action assigned zero baseline probability remains impossible after multiplication by the positive interaction weights; the technology in \eqref{he:eq:root-law}--\eqref{he:eq:pairwise} has this property. The following result is stated only for centipede games.

\begin{corollary}[All-stop outcome in strict constant-sum centipedes]\label{he:corollary:auto-1}
Consider a finite deterministic two-player constant-sum centipede in which backward induction strictly selects $S$ at every decision node. Under a baseline-support-preserving positive interaction-weight technology, every HE follows the backward-induction action at every reached history. Thus the only HE outcome is immediate stopping at node 1.
\end{corollary}

\begin{proof}
Pure backward-induction strategies are security strategies in a deterministic constant-sum centipede. Positive finite interaction weights reweight only actions already in baseline support, so a player using a pure security plan retains its guarantee in the augmented game. The HE auxiliary game therefore has the ordinary centipede value. Strictness then rules out positive probability on any on-path $C$: the opponent can switch to the corresponding pure stopping/security continuation and make that action yield strictly less than the value. Applying the argument successively along the path leaves only $S_1$ as an outcome. Off-path interaction coordinates may remain immaterial.
\end{proof}

For the three-node constant-sum example $(2,3),(1,4),(4,1),(3,2)$, conditional on $C_1$ one has $U_1^{C_1}=[(1-y)+4y(1-x)+3Kxy]/[1+(K-1)xy]$ and $\partial U_1^{C_1}/\partial y=[3(1-x)+2Kx]/\mathcal Z^2>0$. Player 2, who minimizes player 1's payoff, therefore chooses $y=0$; player 1's root payoff becomes $2-t$, so $t=0$. Outcome-wise, HE selects the same first-stop behavior as backward induction.

\subsection{Experimental direction and interpretation}
\label{subsec:experiments}
The experimental literature broadly confirms the same payoff-sensitive direction. The evidence summarized by \citet[Section 4.1]{ismail2025} reports substantially more cooperation in increasing- or positive-sum centipedes and the most noncooperative, SPNE-like behavior in constant-sum designs. In comparative terms, play converges toward cooperative continuation in positive-sum environments but toward early and often first-node stopping in constant-sum environments. The interested reader may consult \citet{ismail2025}.

In Historical Equilibrium, history is not discarded once a node is reached; instead, it provides evidence about the complete plan that generated play. HE evaluates continuation payoffs backward under a fixed structural law, uses rational inference to interpret the realized past, and uses forward inference to connect that past to future behavior. At equilibrium, complete plans, inferred responses, and interaction weights are mutually optimal.

\clearpage
\section{Cautious Backward Induction}
\label{sec:cbi}

\subsection{Introduction}

In finite perfect-information games, payoff ties can support several pure backward-induction (BI) solutions. We study three maximin restrictions on future contingencies. \emph{Actionwise caution} takes the rectangular closure of locally BI-supported actions. \emph{Recursive caution} starts from that BI rectangle but, whenever maximin itself justifies a new non-BI action, adds it to the contingencies faced by predecessors. \emph{Recursive-consistent caution} instead uses only actions selected by the same cautious rule at later nodes. All three exist by backward recursion. The first two can depart from pure BI; recursive-consistent caution always refines it.

The analysis is restricted to finite deterministic games of perfect information and pure strategies. In this class, the possible outputs of pure backward induction coincide with pure subgame-perfect equilibria; see the standard foundations in \citet{Kuhn1953} and \citet{Selten1975}. For related literature, see, e.g., \citet{Tranaes1998,Bonanno2018,BattigalliDeVito2021,Perea2025}.

\subsection{Setup}
A finite deterministic perfect-information game is
\(
G=(I,H,Z,I(\cdot),(A(h))_{h\notin Z},u),
\)
where $I$ is the finite player set, $H$ is a finite rooted tree, $Z$ its terminal histories, $I(h)\in I$ the mover at nonterminal history $h$, $A(h)$ the finite nonempty action set, and $u_i:Z\to\mathbb R$ player $i$'s payoff. The successor after $a\in A(h)$ is $ha$.

An \emph{action correspondence} $D$ assigns a nonempty set $D(h)\subseteq A(h)$ to every decision history. Its reachable terminal set is
\[
\cZ_D(z)=\{z\},\qquad
\cZ_D(h)=\bigcup_{a\in D(h)}\cZ_D(ha).
\]
For $i=I(h)$, the $D$-security payoff and restricted-maximin recommendation are
\begin{equation}\label{caution:eq:restricted}
 m_i^D(h,a)=\min_{z\in\cZ_D(ha)}u_i(z),
 \qquad
 C_D(h)=\argmax_{a\in A(h)}m_i^D(h,a).
\end{equation}
Let $F$ denote the full action correspondence, $F(h)=A(h)$ at every decision history. Standard maximin is $C_F$ and therefore considers every feasible future contingency; write $\cZ(h)=\cZ_F(h)$ for the full feasible terminal set. The concepts below are instances of the same operator: $\CA=C_B$, $\CR=C_R$, and $\CRC=C_{\CRC}$, with the relevant continuation correspondences defined below.\footnote{A distinct profile-consistent alternative minimizes over complete BI continuation profiles rather than nodewise actions. Its objective for action $a$ is exactly $\ell_i(h,a)$, so its recommendation is $\argmax_a\ell_i(h,a)$, the set of actions attaining $\beta_i(h)$; by \eqref{caution:eq:B}, every such action belongs to $B(h)$. This variant is not studied here.}

\subsection*{Pure backward induction with ties}
Let $\cE(h)$ be the set of terminal outcomes induced by pure BI, equivalently pure subgame-perfect equilibrium, in the subgame rooted at $h$. Put $\cE(z)=\{z\}$. If $i=I(h)$, define
\[
 \ell_i(h,a)=\min_{z\in\cE(ha)}u_i(z),
 \qquad
 \beta_i(h)=\max_{a\in A(h)}\ell_i(h,a).
\]

\begin{theorem}[BI recursion and existence]\label{caution:thm:BI}
For every decision history $h$,
\begin{align}
 \cE(h)
 &=\bigcup_{a\in A(h)}
   \{z\in\cE(ha):u_{I(h)}(z)\geq\beta_{I(h)}(h)\}, \label{caution:eq:BI-outcomes}\\
 B(h)
 &=\{a\in A(h):\exists z\in\cE(ha)\text{ with }
                   u_{I(h)}(z)\geq\beta_{I(h)}(h)\}, \label{caution:eq:B}
\end{align}
where $B(h)$ is the set of actions occurring at $h$ in some pure BI solution of the subgame. Hence $\cE(h)$ and $B(h)$ are nonempty and are obtained uniquely by backward recursion.
\end{theorem}

\begin{proof}
Let $i=I(h)$. Proceed by induction on subgame height. If a BI profile chooses $a$ at $h$ and induces $z$, its continuation in every branch $b$ induces some $z_b\in\cE(hb)$; optimality at $h$ implies $u_i(z)\geq u_i(z_b)\geq\min_{y\in\cE(hb)}u_i(y)$, hence $u_i(z)\geq\beta_i(h)$. Conversely, if $z\in\cE(ha)$ meets this threshold, attach after each alternative $b$ a BI continuation attaining $\min_{y\in\cE(hb)}u_i(y)$. Choosing $a$ is then optimal at $h$, yielding a BI profile. Nonemptiness follows because a branch attaining $\beta_i(h)$ has at least one continuation meeting the threshold.
\end{proof}

\subsection{Three cautious rules}
The \emph{BI path closure} $\cZ_B(h)$ contains every outcome obtainable when, at each reached node $x$, the mover independently chooses any action in $B(x)$. Such a path need not be generated by one BI profile: it may switch among local actions justified by different BI solutions.

\begin{definition}[Actionwise caution]\label{caution:def:A}
At a decision history $h$ of player $i$,
\begin{equation}\label{caution:eq:A}
 \CA(h)=\argmax_{a\in A(h)}\min_{z\in\cZ_B(ha)}u_i(z).
\end{equation}
Thus the future contingency set is the fixed rectangular closure of BI-supported actions.
\end{definition}

\begin{definition}[Recursive caution]\label{caution:def:R}
Define $R$ and $\CR$ from the leaves toward the root. Once $R$ is known at all strict descendants of $h$, let
\begin{equation}\label{caution:eq:R}
 \CR(h)=\argmax_{a\in A(h)}\min_{z\in\cZ_R(ha)}u_{I(h)}(z),
 \qquad
 R(h)=B(h)\cup\CR(h).
\end{equation}
Thus BI-supported actions seed the recursion, and any additional action justified by cautious maximin becomes a possible contingency for every predecessor.
\end{definition}

\begin{definition}[Recursive-consistent caution]\label{caution:def:RC}
Put $\cZ_{\CRC}(z)=\{z\}$ at terminal histories. Once $\cZ_{\CRC}$ is known at the successors of $h$, let $i=I(h)$ and set
\begin{equation}\label{caution:eq:RC}
 \CRC(h)=\argmax_{a\in A(h)}\min_{z\in\cZ_{\CRC}(ha)}u_i(z),
 \qquad
 \cZ_{\CRC}(h)=\bigcup_{a\in\CRC(h)}\cZ_{\CRC}(ha).
\end{equation}
At a last decision node this is ordinary best response. Earlier movers maximize the worst payoff generated when all later movers use the same rule.
\end{definition}

The three rules differ only in what remains possible downstream. Actionwise caution keeps the BI rectangle fixed. Recursive caution can expand it by propagating newly cautious actions backward. Recursive-consistent caution instead prunes it by retaining only actions selected by its own recursion.

\begin{theorem}[Existence, nesting, and refinement]\label{caution:thm:existence}
At every decision history, $\CA(h)$, $\CR(h)$, and $\CRC(h)$ are nonempty. Equation~\eqref{caution:eq:R} determines a unique correspondence $R$, while \eqref{caution:eq:RC} determines $\CRC$ uniquely, in each case by backward induction on subgame height. Moreover,
\[
 \cZ_{\CRC}(h)\subseteq\cE(h)\subseteq\cZ_B(h)\subseteq\cZ_R(h)\subseteq\cZ(h),
 \qquad
 \CRC(h)\subseteq B(h).
\]
Hence recursive-consistent caution always refines pure BI. Section~\ref{caution:sec:nested-examples} shows that actionwise and recursive caution need not.
\end{theorem}

\begin{proof}
Finiteness gives nonempty minima and argmax sets, and both recursive definitions depend only on strict descendants. Since $B(h)\subseteq R(h)$, $\cZ_B(h)\subseteq\cZ_R(h)$; every BI outcome follows BI-supported actions, so $\cE(h)\subseteq\cZ_B(h)$. For recursive-consistent caution, use induction on subgame height. Take $a^*\in\CRC(h)$ and $z\in\cZ_{\CRC}(ha^*)$. For each alternative $b$, choose $z_b\in\argmin_{y\in\cZ_{\CRC}(hb)}u_i(y)$, where $i=I(h)$. By induction, $z$ and every $z_b$ are BI outcomes of their successor subgames, while \eqref{caution:eq:RC} implies $u_i(z)\ge u_i(z_b)$ for all $b$. Attaching BI continuations inducing these outcomes makes $a^*$ optimal at $h$. Thus $z\in\cE(h)$ and $a^*\in B(h)$.
\end{proof}

\begin{proposition}[Useful equivalences]\label{caution:prop:equiv}
If $R(x)=B(x)$ at every strict descendant $x$ of $h$, then $\CR(h)=\CA(h)$. If $\cZ_{\CRC}(ha)=\cZ_B(ha)$ for every $a\in A(h)$, then $\CRC(h)=\CA(h)$. In particular, if $B(x)$ is a singleton at every node of a subgame, then
\[
 \CA(x)=\CR(x)=\CRC(x)=B(x)
\]
throughout that subgame.
\end{proposition}
\begin{proof}
Under either hypothesis the relevant branchwise minima are identical. If all BI action sets are singletons, the unique BI continuation is the only path retained by each rule; Theorem~\ref{caution:thm:existence} then gives the claim by backward induction.
\end{proof}

\subsection{Nested examples}\label{caution:sec:nested-examples}
Let $G_k$ denote the subgame with $k$ remaining decision nodes for $k=2,3,4,5$. In Figure~\ref{caution:fig:five-node}, $G_5$, $G_4$, $G_3$, and $G_2$ start at $h_0$, $h_1$, $h_2$, and $h_3$, respectively; $h_4$ is the last decision node. Action strings are written from the root of the relevant subgame toward the terminal node.

\paragraph{The two-node subgame $G_2$.}
At $h_4$, Player 2 is indifferent: $S$ and $C$ both give payoff $3$, so
\[
 B(h_4)=\CA(h_4)=\CR(h_4)=\CRC(h_4)=\{S,C\},\qquad R(h_4)=B(h_4).
\]
The pure BI profiles of $G_2$, rooted at $h_3$, are $SS$ and $CC$, hence $B(h_3)=\{S,C\}$. At $h_3$, Player 1 obtains $2$ from stopping, whereas continuing has floor $\min\{0,3\}=0$. Therefore
\[
 \CA(h_3)=\CR(h_3)=\CRC(h_3)=\{S\}.
\]
Since $S\in B(h_3)$, recursive caution adds nothing new there: $R(h_3)=B(h_3)$.

\paragraph{The three-node subgame $G_3$.}
Its pure BI profiles are $SSS$ and $CCC$, so $B(h_2)=\{S,C\}$. Under actionwise caution, continuing at $h_2$ can give Player 2 payoffs $1,3,3$, whose floor $1$ is below the stopping payoff $2$. Since $R=B$ below $h_2$, recursive caution has the same calculation. Recursive-consistent caution has already selected $S$ at $h_3$, so continuing also gives Player 2 only $1$. Hence
\[
 \CA(h_2)=\CR(h_2)=\CRC(h_2)=\{S\},\qquad R(h_2)=B(h_2).
\]

\begin{figure}[t]
\centering
\resizebox{0.98\textwidth}{!}{%
\begin{tikzpicture}[
    x=1cm,y=1cm,
    every node/.style={font=\small},
    p1/.style={circle,draw=red!75!black,fill=red!75,inner sep=2.2pt,text=white},
    p2/.style={circle,draw=blue!75!black,fill=blue!75,inner sep=2.2pt,text=white},
    edge/.style={line width=.8pt},
    action1/.style={text=red!75!black,font=\small\bfseries},
    action2/.style={text=blue!75!black,font=\small\bfseries}
]
\node[p2,label=above:{$h_0$}] (h0) at (0,0) {2};
\node (z0) at (-2,-1.5) {$z_0=(0,\tfrac12)$};
\node[p1,label=above right:{$h_1$}] (h1) at (2,-1.5) {1};
\node (z1) at (0,-3.0) {$z_1=(\tfrac52,0)$};
\node[p2,label=above right:{$h_2$}] (h2) at (4,-3.0) {2};
\node (z2) at (2,-4.5) {$z_2=(3,2)$};
\node[p1,label=above right:{$h_3$}] (h3) at (6,-4.5) {1};
\node (z3) at (4,-6.0) {$z_3=(2,1)$};
\node[p2,label=above right:{$h_4$}] (h4) at (8,-6.0) {2};
\node (z4) at (6.8,-7.5) {$z_4=(0,3)$};
\node (z5) at (9.2,-7.5) {$z_5=(3,3)$};

\draw[edge] (h0)--node[action2,above left]{$S$}(z0);
\draw[edge] (h0)--node[action2,above right]{$C$}(h1);
\draw[edge] (h1)--node[action1,above left]{$S$}(z1);
\draw[edge] (h1)--node[action1,above right]{$C$}(h2);
\draw[edge] (h2)--node[action2,above left]{$S$}(z2);
\draw[edge] (h2)--node[action2,above right]{$C$}(h3);
\draw[edge] (h3)--node[action1,above left]{$S$}(z3);
\draw[edge] (h3)--node[action1,above right]{$C$}(h4);
\draw[edge] (h4)--node[action2,above left]{$S$}(z4);
\draw[edge] (h4)--node[action2,above right]{$C$}(z5);
\end{tikzpicture}}
\caption{The five-node stop--continue game. The nested games $G_5$, $G_4$, $G_3$, and $G_2$ start at $h_0$, $h_1$, $h_2$, and $h_3$, respectively; $h_4$ is the last decision node.}
\label{caution:fig:five-node}
\end{figure}
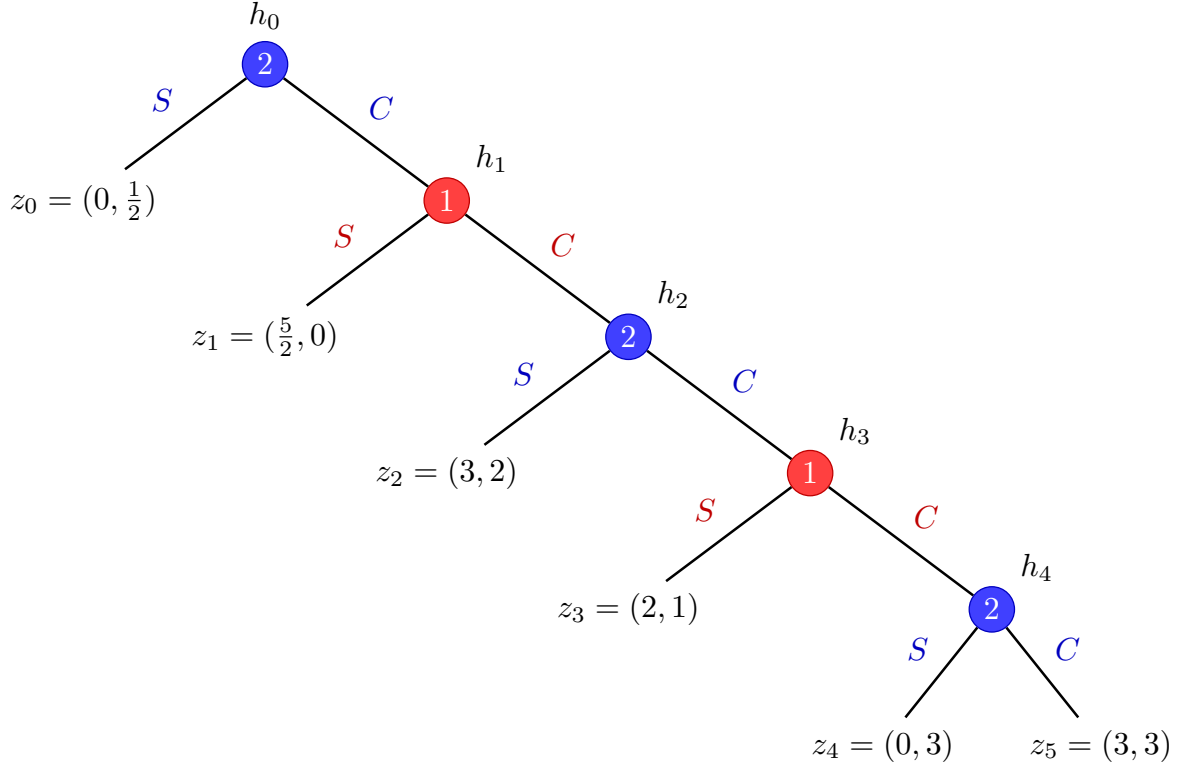

\paragraph{The four-node subgame $G_4$.}
The only pure BI profiles of $G_4$, rooted at $h_1$, are $CSSS$ and $CCCC$, so $B(h_1)=\{C\}$. Actionwise caution nevertheless permits every path assembled from $B(h_2)=B(h_3)=B(h_4)=\{S,C\}$ after initial $C$. The reachable outcomes are $(3,2),(2,1),(0,3),(3,3)$, whose Player-1 payoffs have minimum $0$. Since stopping gives $5/2$,
\[
 \CA(h_1)=\{S\}.
\]
The harmful path $CCCS$ is not one BI profile; it combines locally BI-supported actions under different continuation rationales. Because recursive caution has not yet added any non-BI action at a strict descendant, $R=B$ below $h_1$ and therefore
\[
 \CR(h_1)=\CA(h_1)=\{S\}.
\]
Now recursion matters: $S\notin B(h_1)$, so $R(h_1)=B(h_1)\cup\CR(h_1)=\{C,S\}$ for predecessors.

Recursive-consistent caution behaves differently. It has already selected $S$ at $h_3$ and $h_2$. Thus continuing at $h_1$ leads to $(3,2)$ and gives Player 1 payoff $3>5/2$, so
\[
 \CRC(h_1)=\{C\}=B(h_1).
\]
Thus $G_4$ separates the BI-rectangular rules from the dynamically consistent one.

\paragraph{The five-node game $G_5$.}
At $h_0$, actionwise caution still uses only $B(h_1)=\{C\}$. Conditional on entering, Player 2 can receive $2,1,3,$ or $3$, so its floor is $1>1/2$ and
\[
 \CA(h_0)=\{C\}.
\]
Recursive caution instead uses $R(h_1)=\{S,C\}$. If Player 1 chooses the newly cautious $S$ at $h_1$, Player 2 receives $0$, so the recursive floor from entering is $0<1/2$ and
\[
 \CR(h_0)=\{S\}.
\]
Recursive-consistent caution has $\CRC(h_1)=\{C\}$ and $\CRC(h_2)=\{S\}$, so entering leads to $(3,2)$ and gives Player 2 payoff $2>1/2$; hence
\[
 \CRC(h_0)=\{C\}.
\]
The fifth node therefore separates recursive caution from both actionwise and recursive-consistent caution.

\begin{center}
\small
\begin{tabular}{ccllll}
\toprule
Game & Root & Pure BI profiles & $\CA$ & $\CR$ & $\CRC$\\
\midrule
$G_2$ & $h_3$ & $SS,CC$ & $S$ & $S$ & $S$\\
$G_3$ & $h_2$ & $SSS,CCC$ & $S$ & $S$ & $S$\\
$G_4$ & $h_1$ & $CSSS,CCCC$ & $S$ & $S$ & $C$\\
$G_5$ & $h_0$ & $CCSSS,CCCCC$ & $C$ & $S$ & $C$\\
\bottomrule
\end{tabular}
\end{center}

\begin{proposition}[Refinement and separation]
Actionwise caution and recursive caution need not refine pure BI. Recursive-consistent caution always refines pure BI. The three recommendations can differ pairwise across games.
\end{proposition}
\begin{proof}
In $G_4$, $B(h_1)=\{C\}$ while $\CA(h_1)=\CR(h_1)=\{S\}$ and $\CRC(h_1)=\{C\}$. In $G_5$, $\CA(h_0)=\CRC(h_0)=\{C\}$ whereas $\CR(h_0)=\{S\}$. The general refinement of $\CRC$ is Theorem~\ref{caution:thm:existence}.
\end{proof}

Why not use unrestricted maximin? Suppose Player 1 chooses $S\to(1,0)$ or $C$ leading to Player 2, who chooses $L\to(2,2)$ or $R\to(-1,0)$. Full maximin chooses $S$ because it includes $R$. But $R$ is strictly suboptimal at Player 2's node relative to $L$, so $B=\{L\}$ there; all three rationality-restricted rules choose $C$. The distinctive step is therefore the rationality restriction on contingencies, not maximin itself.

\subsection{Interpretation and further properties}
A hybrid path admitted by actionwise caution is not one BI profile: each action has a local BI justification, but the justification may change with history. Recursive caution keeps this rectangular BI baseline and then propagates any newly maximin-justified non-BI action backward. Recursive-consistent caution imposes a different principle: every future node uses the same cautious rule, so actions not selected by that recursion are pruned from predecessors' contingency sets. This is why $CCCS$ matters for actionwise and recursive caution in $G_4$ but not for recursive-consistent caution.

\subsection{Conclusion}
Cautious backward reasoning is maximin over an endogenous continuation set. Actionwise caution fixes the rectangular closure of BI-supported actions. Recursive caution starts from that rectangle and expands it whenever caution itself justifies a new action. Recursive-consistent caution instead uses a single self-contained rule at every node and therefore prunes future contingencies; it always refines pure BI. In the nested games, all three agree in $G_2$ and $G_3$; $G_4$ separates recursive-consistent caution from the two BI-rectangular rules, while $G_5$ separates recursive caution from the other two.

\clearpage
\section{Simple Nash Equilibrium: Common Behavioral Rules in Extensive-Form Games}
\label{sec:sne}

\subsection{Introduction}

A behavioral strategy assigns a separate lottery to every information set. We study a restriction: a player chooses one lottery over a fixed set of action-coded modes and reuses that lottery at every own information set. A Nash equilibrium in these common lotteries is a \emph{Simple Nash Equilibrium} (SNE). A mixed SNE always exists. We apply SNE to centipede game and finitely repeated prisoner's dilemma.

The related literature includes \citet{Rubinstein1986,AbreuRubinstein1988,Wichardt2008}.

\subsection{Model}

Let $\Gamma$ be a finite extensive-form game with players $N$, terminal histories $Z$, information sets $\I_i$, feasible actions $A(h)$ at $h\in\I_i$, chance-move probabilities, and terminal utilities $u_i:Z\to\R$. The game has perfect recall.

A \emph{behavioral code} for player $i$ consists of a finite mode set $M_i$ and, for every $h\in\I_i$, a map
\[
\alpha_{ih}:M_i\to A(h).
\]
The same mode labels are available at all of player $i$'s information sets, although their action meanings may depend on the node. Let
\[
X_i=\Delta(M_i).
\]
Given $x_i\in X_i$, the induced behavioral strategy is
\begin{equation}
\sigma_i^{x_i}(a\mid h)
=\sum_{m\in M_i}x_i(m)\,\one\{\alpha_{ih}(m)=a\}.
\label{eq:simplebehavior}
\end{equation}
A fresh mode is drawn independently whenever player $i$ moves, but its distribution $x_i$ is identical at every information set. Thus the restriction concerns the behavioral lottery.

The induced payoff is
\begin{equation}
U_i(x)=\sum_{z\in Z}\Pr_{\Gamma}(z\mid \sigma^x)u_i(z),
\qquad x\in X\equiv\prod_j X_j.
\label{eq:payoff}
\end{equation}
The associated \emph{simple game} is the strategic-form game $(N,(X_i)_i,(U_i)_i)$.

\begin{definition}[Simple Nash equilibrium]
A profile $x^*\in X$ is a \emph{Simple Nash Equilibrium} if, for every player $i$,
\[
U_i(x_i^*,x_{-i}^*)\ge U_i(x_i,x_{-i}^*)
\qquad\forall x_i\in X_i.
\]
\end{definition}

SNE is therefore ordinary Nash equilibrium in the induced game over common behavioral rules. If $S_i^\alpha=\{\sigma_i^{x_i}:x_i\in X_i\}$ denotes the induced restricted behavioral-strategy set, then an SNE is exactly a profile with no profitable deviation inside $S_i^\alpha$. Hence
\[
\operatorname{NE}(\Gamma)\cap\prod_iS_i^\alpha
\subseteq \operatorname{SNE}(\Gamma,\alpha),
\]
and the inclusion can be strict.

A useful binary case has two modes, $F$ and $O$. Mode $F$ follows a reference rule $r_i$, while $O$ follows its opposite. Then one number $p_i=\Pr(F)$ governs all of player $i$'s nodes. Literal common mixing is the special case in which $F$ and $O$ correspond to the same two action labels at every node. In repeated games, $F$ may instead represent a history-dependent rule such as Grim trigger.

\subsection{Structure and existence}

Let $d_i(z)$ be the number of player $i$'s decision nodes on terminal history $z$, and $d_i=\max_z d_i(z)$.

\begin{proposition}
For every player $j$, $U_j(x)$ is a multivariate polynomial, with degree in $x_i$ at most $d_i$.
\end{proposition}

\begin{proof}
The realization probability of each terminal history is a product of chance probabilities and behavioral action probabilities. By \eqref{eq:simplebehavior}, each behavioral probability is linear in the relevant common lottery. Player $i$ contributes at most $d_i(z)$ such factors along history $z$. Summing over terminal histories preserves the degree bound.
\end{proof}

The polynomial structure distinguishes SNE from unrestricted behavioral optimization, where payoff is multi-affine in separate information-set probabilities. Identifying those probabilities can destroy quasiconcavity and therefore pure existence.

The main behavioral implication is easiest to see in a binary code. Let $b_{ih}$ be the unrestricted probability of following the reference action at information set $h$, and let $\widetilde U_i((b_{jk})_{j,k})$ be expected payoff in the full behavioral coordinates. Under SNE, $b_{ih}=p_i$ for every $h\in\I_i$.

\begin{proposition}
At every differentiability point,
\begin{equation}
\frac{\partial U_i}{\partial p_i}
=\sum_{h\in\I_i}
\frac{\partial \widetilde U_i}{\partial b_{ih}}.
\label{eq:aggregation}
\end{equation}
Hence an interior SNE may satisfy $\partial U_i/\partial p_i=0$ even though the individual nodewise marginal incentives are nonzero and have opposite signs.
\end{proposition}

\begin{proof}
Since $U_i(p)=\widetilde U_i(b(p))$ and $b_{ih}=p_i$, the result follows immediately from the chain rule.
\end{proof}

Equation \eqref{eq:aggregation} is an important mechanism. A prescribed action may be strictly inferior at one reached node, but changing the common parameter also changes behavior at other nodes. SNE therefore requires rule-level optimality.

Existence requires distinguishing a behavioral SNE from an ex ante mixture over common lotteries. A \emph{mixed SNE} is a Nash equilibrium in probability measures over the compact sets $X_i$.

\begin{theorem}[Existence]
Every finite coded extensive-form game has a mixed SNE. A behavioral SNE exists under any of the following sufficient conditions:
\begin{enumerate}[label=(\roman*)]
\item $U_i(\cdot,x_{-i})$ is quasiconcave on $X_i$ for every $i$ and $x_{-i}$;
\item every player moves at most once along every terminal history;
\item in the binary case $X_i=[0,1]$, the induced simple game has increasing differences.
\end{enumerate}
Under (iii), the equilibrium set has least and greatest elements.
\end{theorem}

\begin{proof}
Continuity follows from the polynomial-payoff proposition. Compactness and continuity give mixed-strategy existence in the induced continuous game. Under (i), standard fixed-point existence applies because $X_i$ is compact and convex. Under (ii), each own payoff is affine in $x_i$, hence quasiconcave. Under (iii), monotone best responses and the lattice fixed-point argument for supermodular games apply.
\end{proof}

Behavioral existence can fail because the common parameter may enter the same player's payoff nonlinearly. The common-rule restriction can create genuinely nonconvex best-response problems. 

\subsection{Two applications}

\subsection{Centipede}

Consider a three-decision centipede in Figure \ref{he:eq:centipede-tree-SNE}. 

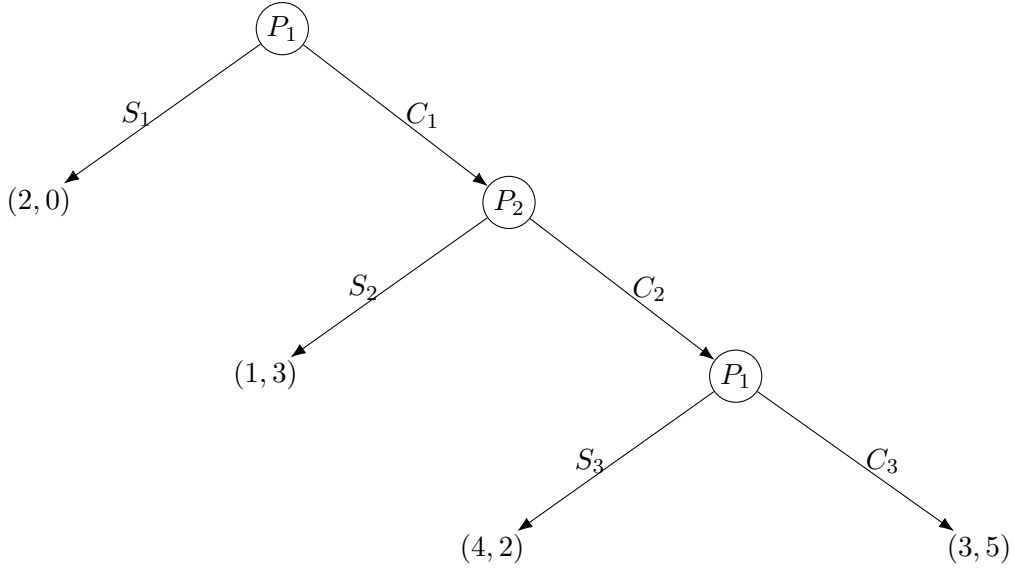
\begin{figure}[ht]
\centering
\begin{tikzpicture}[
  node distance=18mm and 25mm,
  decision/.style={circle,draw,inner sep=1.7pt},
  terminal/.style={rectangle,draw=none,inner sep=1pt}
]
\node[decision] (p1a) {$P_1$};
\node[terminal,below left=of p1a] (z1) {$(2,0)$};
\node[decision,below right=of p1a] (p2) {$P_2$};
\node[terminal,below left=of p2] (z2) {$(1,3)$};
\node[decision,below right=of p2] (p1b) {$P_1$};
\node[terminal,below left=of p1b] (z3) {$(4,2)$};
\node[terminal,below right=of p1b] (z4) {$(3,5)$};
\draw[-{Latex[length=2mm]}] (p1a) -- node[left] {$S_1$} (z1);
\draw[-{Latex[length=2mm]}] (p1a) -- node[right] {$C_1$} (p2);
\draw[-{Latex[length=2mm]}] (p2) -- node[left] {$S_2$} (z2);
\draw[-{Latex[length=2mm]}] (p2) -- node[right] {$C_2$} (p1b);
\draw[-{Latex[length=2mm]}] (p1b) -- node[left] {$S_3$} (z3);
\draw[-{Latex[length=2mm]}] (p1b) -- node[right] {$C_3$} (z4);
\end{tikzpicture}
\caption{The three-decision centipede.}
\label{he:eq:centipede-tree-SNE}
\end{figure}

Player 1 uses the same pass probability $p$ at both of her nodes; Player 2 uses pass probability $q$. Expected payoffs are
\begin{align}
U_1(p,q)&=2(1-p)+p(1-q)+4pq(1-p)+3p^2q
=2-p+3pq-p^2q,\label{eq:cent1}\\
U_2(p,q)&=3p(1-q)+2pq(1-p)+5p^2q
=3p-pq+3p^2q.\label{eq:cent2}
\end{align}

\begin{proposition}[Centipede SNE]
The set of SNE is
\[
\left\{(0,q):0\le q\le\frac13\right\}
\cup
\left\{\left(\frac13,\frac37\right),(1,1)\right\}.
\]
\end{proposition}

\begin{proof}
From \eqref{eq:cent1},
\[
\frac{\partial U_1}{\partial p}=-1+3q-2pq,
\]
so $BR_1(q)=\{0\}$ for $q\le1/3$ and $BR_1(q)=\{(3q-1)/(2q)\}$ for $q>1/3$. From \eqref{eq:cent2},
\[
U_2(p,q)=3p+pq(3p-1),
\]
so Player 2 chooses $q=0$ for $0<p<1/3$, is indifferent at $p\in\{0,1/3\}$, and chooses $q=1$ for $p>1/3$. Intersecting best responses gives the stated set.
\end{proof}

The full-pass equilibrium shows why SNE differs from nodewise reasoning. At Player 1's final node, Take yields $4$ and Pass yields $3$, so Pass is strictly locally inferior. But against $q=1$,
\[
U_1(p,1)=2+2p-p^2,
\]
which is maximized at $p=1$. Lowering $p$ improves Player 1's action at the final node but simultaneously raises the probability of taking $(2,0)$ at the first node. The whole rule is optimal even though one component is not.

\subsection{Finitely repeated prisoner's dilemma}

Consider an $H$-period repetition of
\[
\begin{array}{c|cc}
 & C & D\\ \hline
C&(R,R)&(S,T)\\
D&(T,S)&(P,P)
\end{array}
\qquad T>R>P>S.
\]
Take the reference mode to be Grim: cooperate while no defection has occurred and defect forever after the first defection. The opposite mode defects in the good state and cooperates in the punishment state. Player $i$ follows Grim at every history with the same probability $p_i$.

Fix the opponent at full Grim. At a good-state node, departing from Grim raises current payoff by $T-R$. Once punishment begins, a player who follows Grim with probability $x$ obtains per-period payoff
\[
V(x)=xP+(1-x)S\le P.
\]
Thus the loss from entering punishment, evaluated at full Grim, is $R-P$.

\begin{theorem}[Finite-horizon cooperation]
Full Grim, $p_1=p_2=1$, is an SNE if and only if
\begin{equation}
(H-1)(R-P)\ge 2(T-R).
\label{eq:pd}
\end{equation}
\end{theorem}

\begin{proof}
Suppose one player uses Grim $x<1$ while the opponent uses one. The first departure from Grim occurs in period $t$ with probability $x^{t-1}(1-x)$. Relative to full Grim, that history gives current gain $T-R$ and then $H-t$ punishment periods. Let
\[
L_H(x)=\frac{\sum_{t=1}^H x^{t-1}(H-t)}{\sum_{t=1}^H x^{t-1}}.
\]
Because the weights $x^{t-1}$ place relatively more mass on early periods, $L_H(x)\ge(H-1)/2$, with equality as $x\uparrow1$. Since $V(x)\le P$, any deviation has payoff gain bounded above by a positive factor times
\[
(T-R)-(R-P)L_H(x).
\]
Hence no deviation is profitable if \eqref{eq:pd} holds. Conversely, if \eqref{eq:pd} fails, values of $x$ sufficiently close to one make this expression positive, so full Grim is not a best response.
\end{proof}

This result isolates the SNE mechanism. Standard backward induction permits a last-period defection without changing earlier behavior. Under SNE, that deviation is unavailable: lowering Grim in the last period necessarily lowers it at all earlier good-state nodes as well. Full cooperation is therefore sustained exactly when the horizon-weighted punishment loss offsets the one-period temptation gain.

For the payoffs $(T,R,P,S)=(5,3,1,0)$, condition \eqref{eq:pd} becomes $2(H-1)\ge4$, so full Grim is an SNE for every $H\ge3$. The same trigger-rule argument extends directly to finite public-goods games. With $n$ players and payoff $w-c_i+\alpha\sum_j c_j$, full conditional contribution is an SNE exactly when
\[
(H-1)(n\alpha-1)\ge2(1-\alpha).
\]
Thus the mechanism is not specific to two-player dilemmas.

As discussed in subsection \ref{subsec:experiments}, SNE is consistent with continuation in centipede games and cooperation in repeated PDs.

\clearpage
\section{The Maastricht Paradox: Equilibrium Unraveling and Pure Optimin in the Maas Game}
\label{sec:maastricht}

\noindent
In this section, we introduce the \emph{Maas game}, named after Maastricht, where it was conceived, is a dynamic game of common-resource exploitation with features of the Prisoner's Dilemma and the tragedy of the commons. Its central contrast is that, under simple sufficient conditions, every Nash equilibrium of every current subgame recommends immediate liquidation, and the unique subgame-perfect equilibrium (SPNE) prescribes immediate liquidation at every active history. By contrast, pure optimin can prescribe preservation for a very long time, with liquidation delayed substantially as players become more patient.

We call this contrast the \emph{Maastricht paradox}. The equilibrium result is driven by unraveling: because liquidation ends the relationship and cannot be punished afterward, the eventual incentive to preempt propagates backward to the current round. Thus immediate equilibrium liquidation can arise even in an infinite-horizon game with a rapidly growing resource, whereas pure optimin can prescribe sustained preservation under the same underlying payoff structure.

\subsection{The Maas game}

The game is motivated by a tragedy-of-the-commons interpretation. Choosing $C$ means sustainable use: both players receive a current flow and, if both choose $C$, the common resource survives. Choosing $D$ means liquidation or over-exploitation: the resource is exhausted and the interaction ends. The key feature is that liquidation cannot be punished afterward.

There are two players, a discount factor $\delta\in(0,1)$, and stages $t=1,2,\ldots,N$, where $N\in\mathbb N\cup\{\infty\}$. Stage $t$ is reached only if both players chose $C$ at every earlier stage. At an active stage, the players simultaneously choose $D$ or $C$.

Let $g_t>0$ denote the scale of the resource's liquidation value at stage $t$, and let $r_t\ge0$ denote the sustainable-flow payoff obtained by each player when both choose $C$. Measured from stage $t$, the stage game is
\begin{equation}
G_t=
\begin{array}{c|cc}
 & D & C\\ \hline
D & (2g_t,2g_t) & (5g_t,g_t)\\[1mm]
C & (g_t,5g_t) & (r_t,r_t)+\delta G_{t+1}.
\end{array}
\label{eq:generalgame}
\end{equation}
Any occurrence of $D$ ends the game. If $N<\infty$, mutual cooperation at $N$ yields $(r_N,r_N)$ and the game ends. In the infinite game, the discounted flow sums used below are assumed finite.

A pure strategy is payoff-equivalent to a first stopping time: cooperate as long as the interaction remains active and choose $D$ at the first prescribed stopping stage. A complete extensive-form strategy still contains prescriptions at all active histories, including histories that are off the equilibrium path. For calculations from stage $1$, define
\begin{equation}
R_t:=\sum_{s=1}^{t-1}\delta^{s-1}r_s,
\qquad
X_t:=\delta^{t-1}g_t,
\label{eq:RXgeneral}
\end{equation}
and
\begin{equation}
\ell_t:=R_t+X_t,
\qquad
m_t:=R_t+2X_t,
\qquad
h_t:=R_t+5X_t.
\label{eq:lmhgeneral}
\end{equation}
Thus if player 1 stops first at $t$, the payoff pair is $(h_t,\ell_t)$; if both stop at $t$, it is $(m_t,m_t)$.

\subsection{General equilibrium conditions}

For a fixed $\delta$, consider the following conditions:
\begin{align}
\text{(P)}\qquad &5g_t>r_t+2\delta g_{t+1}
&&\text{for every relevant }t,
\label{eq:P}\\
\text{(T1)}\qquad &\exists q\in(0,1),\ \exists T_0<\infty:\quad
\delta g_{t+1}\le qg_t
&&\text{for all }t\ge T_0,
\label{eq:T1}\\
\text{(T2)}\qquad &
F_t(\delta):=\sum_{k=0}^{\infty}\delta^k r_{t+k}<\infty,
\qquad
\frac{F_t(\delta)}{g_t}\longrightarrow0.
\label{eq:T2}
\end{align}
Condition \eqref{eq:P} is the one-step preemption inequality. If both players are known to stop next period, then cooperating today gives $r_t+2\delta g_{t+1}$, whereas unilaterally liquidating today gives $5g_t$. Conditions \eqref{eq:T1}--\eqref{eq:T2} say that sufficiently far in the future, discounted growth of the liquidation prize is eventually contracting and the entire sustainable-flow tail is negligible relative to current liquidation value.

\begin{lemma}[Tail domination]
\label{lem:tail}
Suppose \eqref{eq:T1}--\eqref{eq:T2} hold. For every starting date $t$, there exists a finite $T\ge t$ such that, in the subgame beginning at $T$, stopping at $T$ weakly dominates every later stopping time, including never stopping, and the domination is strict whenever the opponent stops at or after $T$ or never stops.
\end{lemma}

\begin{proof}
Fix the starting date $t$. Choose $T\ge\max\{t,T_0\}$ so large that
\[
F_T(\delta)<5(1-q)g_T.
\]
By iterating \eqref{eq:T1}, for every $k\ge1$,
\[
\delta^k g_{T+k}\le q^k g_T\le qg_T.
\]
If the opponent stops at $T$, stopping at $T$ gives $2g_T$ rather than $g_T$, so stopping at $T$ is strictly better. If the opponent stops at $T+k$ for some $k\ge1$, any plan that waits beyond $T$ can obtain at most the entire cooperative-flow tail plus the most favorable terminal prize, and hence has payoff, measured from $T$, at most
\[
F_T(\delta)+5\delta^k g_{T+k}
\le F_T(\delta)+5qg_T
<5g_T.
\]
Stopping at $T$ against an opponent who continues gives exactly $5g_T$, so immediate stopping is strictly better.

If instead the opponent never stops, a player who also never stops receives exactly $F_T(\delta)<5g_T$, and any player who stops later at $T+k$ is covered by the preceding bound. Thus stopping at $T$ is also strictly better against never stopping. This proves the claim.
\end{proof}

\begin{theorem}[General equilibrium unraveling]
\label{thm:general}
Fix $\delta\in(0,1)$.
\begin{enumerate}[label=(\roman*)]
\item For a finite horizon $N$, suppose \eqref{eq:P} holds for $t<N$ and $5g_N>r_N$. Then, in every Nash equilibrium of every active subgame, both players choose $D$ with probability one in the current round. Consequently, the unique SPNE prescribes $D$ at every active history.
\item For the infinite horizon, suppose \eqref{eq:P}, \eqref{eq:T1}, and \eqref{eq:T2} hold. Then the same conclusion holds: every Nash equilibrium of every active subgame chooses $D$ with probability one in the current round, and the unique SPNE prescribes $D$ at every active history.
\end{enumerate}
Hence, whenever the hypotheses hold from stage $1$, equilibrium liquidation occurs immediately.
\end{theorem}

\begin{proof}
Consider first a finite active subgame beginning at stage $t$. A mixed strategy is payoff-equivalent within this subgame to a distribution over first stopping times, with $\infty$ denoting no liquidation before the horizon. Let $a$ and $b$ be the largest stopping times in the two players' supports, ordering $\infty$ after $N$. If $a>b$, then stopping at $b$ gives player 1 the same payoff as stopping at $a$ against every opponent stop before $b$, and a strictly larger payoff when the opponent stops at $b$, because $m_b>\ell_b$. Hence $a$ cannot be a best response. Symmetrically, $b>a$ is impossible. Thus the two maximal support points coincide; call the common point $s$.

If $t<s\le N$, compare stopping at $s$ with stopping at $s-1$. Against an opponent stop before $s-1$, the two plans give the same payoff. Against an opponent stop at $s-1$, stopping at $s-1$ gives $m_{s-1}>\ell_{s-1}$. Against an opponent stop at $s$,
\[
h_{s-1}-m_s
=\delta^{s-2}\bigl(5g_{s-1}-r_{s-1}-2\delta g_s\bigr)>0
\]
by \eqref{eq:P}. Since the opponent assigns positive probability to $s$, stopping at $s$ cannot be a best response. If $s=\infty$, stopping at $N$ is weakly better against every finite opponent stop and strictly better against $\infty$, because
\[
h_N-W_N=\delta^{N-1}(5g_N-r_N)>0.
\]
Therefore $s=t$, so both players stop in the current round with probability one. This proves part (i).

For the infinite game, fix an active subgame beginning at $t$ and choose $T\ge t$ as in Lemma~\ref{lem:tail}. Let $(\mu_1,\mu_2)$ be a Nash equilibrium in stopping-time distributions. If, say, $\mu_2$ assigns positive probability to stopping at or after $T$, including never stopping, then every stopping time of player 1 later than $T$, as well as never stopping, is strictly worse than stopping at $T$ on an event of positive probability and no better otherwise. Hence player 1's support is contained in the finite set $\{t,\ldots,T\}$. Let $a$ be its largest support point. Any stopping time of player 2 later than $a$ is then strictly worse than stopping at $a$: the two plans coincide against every opponent stop before $a$, while stopping at $a$ gives $m_a>\ell_a$ when player 1 stops at $a$, an event with positive probability. Hence player 2's support is finite as well.

If instead $\mu_2$ assigns no probability to stopping at or after $T$ and no probability to never stopping, then its support is already contained in $\{t,\ldots,T-1\}$, and the same largest-support argument makes player 1's support finite. Thus every infinite-horizon equilibrium has finite stopping-time supports. The finite maximal-support argument above then applies and implies stopping at the current date $t$ with probability one.

The strategy profile prescribing $D$ at every active history is indeed a SPNE, since against current $D$, choosing $D$ yields $2g_t>g_t$. Since the current action is uniquely pinned down in every active subgame, this SPNE is unique.
\end{proof}

The infinite-horizon argument is the counterpart of ``backward induction from infinity.'' There is no literal last stage but Lemma~\ref{lem:tail} solves an entire sufficiently distant infinite tail at once, after which condition \eqref{eq:P} propagates liquidation backward through the finitely many preceding stages.

\begin{corollary}[A discount-factor-independent sufficient condition]
\label{cor:uniform}
Suppose the sequences $(g_t,r_t)$ satisfy
\begin{align}
&r_t+2g_{t+1}<5g_t &&\text{for all }t,\label{eq:U1}\\
&\frac{g_{t+1}}{g_t}\longrightarrow1,\label{eq:U2}\\
&\frac{\sum_{k=0}^{\infty}\delta^k r_{t+k}}{g_t}\longrightarrow0
&&\text{for every }\delta\in(0,1).\label{eq:U3}
\end{align}
Then the conclusion of Theorem~\ref{thm:general} holds for every $\delta\in(0,1)$, for every finite horizon and for the infinite horizon.
\end{corollary}

\begin{proof}
Condition \eqref{eq:U1} implies \eqref{eq:P} for every $\delta<1$ and also implies $5g_t>r_t$. From \eqref{eq:U2}, for any fixed $\delta<1$, choose $q\in(\delta,1)$. Then, for all sufficiently large $t$,
\[
\delta\frac{g_{t+1}}{g_t}\le q,
\]
which is \eqref{eq:T1}. Condition \eqref{eq:U3} is exactly \eqref{eq:T2}.
\end{proof}

\subsection{A specification}
\label{sec:special}

The general result does not depend on any particular functional form. To make the Maastricht paradox clear while retaining the theorem for every $\delta<1$, consider
\begin{equation}
g_t=e^{\sqrt t},
\qquad
r_t=t.
\label{eq:special}
\end{equation}
The stock value grows faster than every fixed polynomial, while its proportional growth factor tends to one:
\[
\frac{g_{t+1}}{g_t}=e^{\sqrt{t+1}-\sqrt t}\longrightarrow1.
\]
Moreover $t<e^{\sqrt t}$ for every $t\ge1$. One way to see this is to write $x=\sqrt t$ and note that $x^2e^{-x}\le4e^{-2}<1$. Also,
\[
\sqrt{t+1}-\sqrt t
=\frac{1}{\sqrt{t+1}+\sqrt t}
\le\sqrt2-1,
\]
so
\[
\frac{g_{t+1}}{g_t}
\le e^{\sqrt2-1}<2.
\]
Hence
\[
r_t+2g_{t+1}
=t+2g_{t+1}
<g_t+4g_t
=5g_t,
\]
so \eqref{eq:U1} holds. Finally,
\[
F_t(\delta)
=\sum_{k=0}^{\infty}\delta^k(t+k)
=\frac{t}{1-\delta}+\frac{\delta}{(1-\delta)^2},
\]
and therefore $F_t(\delta)/e^{\sqrt t}\to0$. Corollary~\ref{cor:uniform} applies.

\begin{corollary}[Maas-game SPNE for the specification]
For \eqref{eq:special}, for every horizon $N\in\mathbb N\cup\{\infty\}$ and every $\delta\in(0,1)$, every Nash equilibrium of every active subgame chooses $D$ in its current round with probability one, and the unique SPNE prescribes $D$ at every active history. From stage $1$, the realized payoff is $(2e,2e)$.
\end{corollary}

A normal-form Nash equilibrium of the full dynamic game need not prescribe equilibrium play in active subgames that are off its equilibrium path. Thus different Nash equilibrium strategy profiles may differ at such unreached active histories. The stronger statement established above is that, when any active subgame is considered on its own, every Nash equilibrium of that subgame chooses $D$ immediately; subgame perfection therefore selects the unique complete strategy profile prescribing $D$ at every active history.

\subsection{Pure-strategy optimin for the specification}

We now retain \eqref{eq:special} and restrict the game to pure strategies. Let $S_i^P$ denote player $i$'s pure-strategy space, represented payoff-equivalently by stopping dates. For finite $N$,
\[
S_i^P=\{1,\ldots,N,\infty\},
\]
where $\infty$ means no liquidation before the horizon; for $N=\infty$,
\[
S_i^P=\mathbb N\cup\{\infty\}.
\]
The term \emph{pure optimin} refers to the optimin on the restricted game $S_1^P\times S_2^P$.

Following the pure-strategy restriction of the optimin criterion in \citet{ismail2025}, for a pure profile $s=(s_i,s_j)\in S_i^P\times S_j^P$, define player $j$'s admissible pure responses by
\begin{equation}
B_j^P(s)
:=
\{s_j\}
\cup
\left\{
 s'_j\in S_j^P:
 u_j(s_i,s'_j)>u_j(s)
\right\},
\label{eq:pure-response-set}
\end{equation}
and define player $i$'s pure-strategy performance by
\begin{equation}
\pi_i^P(s)
:=
\inf_{s'_j\in B_j^P(s)}u_i(s_i,s'_j).
\label{eq:performance}
\end{equation}
A profile $s\in S_1^P\times S_2^P$ is a \emph{pure optimin} if its performance vector
\[
\pi^P(s):=(\pi_1^P(s),\pi_2^P(s))
\]
is Pareto-undominated among the performance vectors generated by pure profiles. There may be several pure optimins, just as there may be several Nash equilibria.

For \eqref{eq:special}, define
\begin{equation}
A_t:=\sum_{s=1}^{t-1}\delta^{s-1}s,
\qquad
X_t:=\delta^{t-1}e^{\sqrt t},
\label{eq:AXspecial}
\end{equation}
and
\begin{equation}
\ell_t:=A_t+X_t,
\qquad
m_t:=A_t+2X_t,
\qquad
h_t:=A_t+5X_t.
\label{eq:lmhspecial}
\end{equation}
If $m<n$, the stopping profile $(m,n)$ yields $(h_m,\ell_m)$, while $(t,t)$ yields $(m_t,m_t)$. If neither player liquidates, define
\[
W_N:=\sum_{t=1}^{N}\delta^{t-1}t
\quad(N<\infty),
\qquad
W_\infty:=\frac{1}{(1-\delta)^2}.
\]

For a finite stopping time $t$, let
\begin{align}
a_t&:=\min\left\{m_t,
\inf_{\substack{k<t\\h_k>m_t}}\ell_k\right\},\label{eq:at}\\
b_t&:=\min\left\{m_t,
\inf_{\substack{k<t\\h_k>\ell_t}}\ell_k\right\},\label{eq:bt}\\
c_t&:=\min\left\{\ell_t,
\inf_{\substack{k<t\\h_k>h_t}}\ell_k\right\},\label{eq:ct}
\end{align}
where the infimum of the empty set is $+\infty$. Also define
\begin{equation}
a_\infty:=\min\left\{W_N,
\inf_{\substack{k\le N\\h_k>W_N}}\ell_k\right\}
\label{eq:ainf}
\end{equation}
for finite $N$, and use $k\in\mathbb N$ and $W_\infty$ when $N=\infty$.

\begin{proposition}[Exact pure-optimin reduction]
\label{prop:optimin}
For the game under \eqref{eq:special},
\begin{equation}
\pi^P(t,t)=(a_t,a_t),
\qquad
\pi^P(m,n)=(b_m,c_m)\quad\text{for every }m<n,
\label{eq:allperf}
\end{equation}
and symmetrically
\[
\pi^P(n,m)=(c_m,b_m)\quad\text{for every }m<n.
\]
Consequently, the set of pure-profile performance vectors is exactly
\begin{equation}
\mathcal P_N
:=
\{(a_t,a_t),(b_t,c_t),(c_t,b_t):1\le t\le N\}
\cup\{(a_\infty,a_\infty)\}
\label{eq:PN}
\end{equation}
for finite $N$, with $t+1=\infty$ understood at the terminal date when an adjacent representative is desired. For $N=\infty$, replace $1\le t\le N$ by $t\in\mathbb N$. The pure optimins are exactly the pure profiles whose performance vectors lie on the Pareto frontier $\operatorname{PF}(\mathcal P_N)$. In particular, all asymmetric profiles with the same earlier stopping date have the same pure-strategy performance vector.
\end{proposition}

\begin{proof}
At $(t,t)$, an opponent who stops later receives only $\ell_t<m_t$ and therefore does not profit. An earlier pure stop $k<t$ is profitable exactly when $h_k>m_t$ and leaves the nondeviator with $\ell_k$. Together with obedience, this gives
\[
\pi^P(t,t)=(a_t,a_t).
\]

Now take $(m,n)$ with $m<n$. Consider first the performance of the early player. The later player initially receives $\ell_m$. Matching the early stop at $m$ is strictly profitable and leaves the early player with $m_m$. An earlier stop $k<m$ is profitable exactly when $h_k>\ell_m$ and leaves the early player with $\ell_k$. Any later pure deviation by the later player leaves the game already terminated at $m$ and hence cannot improve that player's payoff. Therefore the early player's performance is $b_m$, independently of $n$.

For the later player's performance, the early player initially receives $h_m$. A deviation by the early player to $k<m$ is profitable exactly when $h_k>h_m$ and leaves the later player with $\ell_k$, generating the term in $c_m$. It remains to show that no profitable deviation by the early player to a date $k>m$ can lower the later player's payoff below $\ell_m$.

If $m<k<n$ and the deviation is profitable, then $h_k>h_m$. Since $A_k\ge A_m$ and
\[
\ell_s=\frac{h_s+4A_s}{5},
\]
it follows that
\[
\ell_k
=\frac{h_k+4A_k}{5}
>\frac{h_m+4A_m}{5}
=\ell_m.
\]
If the early player deviates to $n$, profitability requires $m_n>h_m$, and the later player then receives $m_n>h_m>\ell_m$. If the early player deviates to a finite date $k>n$, the early player receives $\ell_n$; if that deviation is profitable, then $\ell_n>h_m>\ell_m$, while the later player receives $h_n>\ell_n>\ell_m$. Finally, if the deviation is to never stopping, profitability likewise requires the deviator's resulting payoff to exceed $h_m$, and the later player's payoff cannot fall below $\ell_m$. Thus only deviations to dates $k<m$ can lower the later player's payoff below obedience, and its performance is exactly $c_m$. This proves \eqref{eq:allperf}; the transpose follows by symmetry.

Hence every pure profile generates one of the vectors in $\mathcal P_N$, and every vector in $\mathcal P_N$ is generated by a pure profile. The characterization of pure optimins by the Pareto frontier follows directly from the definition.
\end{proof}

\begin{proposition}[Existence of a pure optimin]
\label{prop:existence}
For every finite horizon, a pure optimin exists. For the infinite-horizon special specification \eqref{eq:special}, a pure optimin also exists.
\end{proposition}

\begin{proof}
For finite $N$, the pure-strategy space is finite, so there are finitely many pure-profile performance vectors and their Pareto frontier is nonempty.

For $N=\infty$, equip $S^P:=\mathbb N\cup\{\infty\}$ with the topology in which $t\to\infty$; equivalently, use the metric
\[
d(m,n)=\left|\frac1m-\frac1n\right|,
\qquad
\frac1\infty:=0.
\]
Then $S^P$ is compact. Since
\[
A_t\longrightarrow W_\infty,
\qquad
X_t=\delta^{t-1}e^{\sqrt t}\longrightarrow0,
\]
the pure stopping-game payoff function is continuous on $S^P\times S^P$.

Each performance component $\pi_i^P$ is upper semicontinuous. Let $s^n\to s$ and fix $\varepsilon>0$. Choose an admissible pure response at $s$ whose payoff to the evaluated player is within $\varepsilon$ of $\pi_i^P(s)$. If that response is obedience, use obedience at $s^n$. If it is a strictly profitable pure deviation, continuity preserves strict profitability for all sufficiently large $n$. Continuity of payoffs then gives
\[
\limsup_{n\to\infty}\pi_i^P(s^n)
\le \pi_i^P(s)+\varepsilon.
\]
Letting $\varepsilon\downarrow0$ proves upper semicontinuity. Therefore $\pi_1^P+\pi_2^P$ attains a maximum on the compact pure-strategy space. Any maximizer is Pareto-undominated in pure-strategy performance and hence is a pure optimin.
\end{proof}

\subsection{Comparison and examples}

Immediate stopping $(1,1)$ has pure-strategy performance $(2e,2e)$. It is a pure optimin for sufficiently impatient players, but need not remain one as patience increases. Direct evaluation of \eqref{eq:at}--\eqref{eq:PN} in the infinite-horizon game gives the following symmetric pure optimins; asymmetric pure optimins may coexist.

\begin{center}
\begin{tabular}{@{}c c c c@{}}
\toprule
$\delta$ & Every NE / SPNE outcome & Symmetric pure-optimin times & Common performance\\
\midrule
$0.20$ & liquidate at $t=1$ & $t=1$ & $5.44$\\
$0.80$ & liquidate at $t=1$ & $t=19,\ldots,24$ & $12.41$\\
$0.99$ & liquidate at $t=1$ & $t=2457,\ldots,2493$ & $2.572\times10^{10}$\\
\bottomrule
\end{tabular}
\end{center}

Thus patience never changes the Nash or SPNE recommendation under \eqref{eq:special}, but it can move pure-optimin liquidation thousands of periods into the future.

A useful high-patience benchmark comes from maximizing the discounted stock term
\[
X_t=\delta^{t-1}e^{\sqrt t}.
\]
Treating $t$ as continuous,
\[
\log X_t=(t-1)\log\delta+\sqrt t,
\]
so the unconstrained continuous maximizer is
\begin{equation}
t_X^{\mathrm{cont}}
=\frac{1}{4(-\log\delta)^2}
\sim\frac{1}{4(1-\delta)^2}
\qquad(\delta\uparrow1).
\label{eq:tx}
\end{equation}
For the constrained domain $t\ge1$, the continuous maximizer is $\max\{1,t_X^{\mathrm{cont}}\}$, and the integer maximizer is attained at a neighboring integer. For $\delta=0.99$, \eqref{eq:tx} gives approximately $2475$, close to the exact symmetric pure-optimin interval above. Equation \eqref{eq:tx} is a high-patience scale for delayed liquidation.

For a finite horizon, the same pure-optimin formulas apply after restricting stopping times to the available dates and replacing $W_\infty$ by $W_N$. The horizon can truncate or otherwise alter the pure-optimin frontier, but under the equilibrium conditions of Theorem~\ref{thm:general} it never changes the Nash/SPNE conclusion: every active subgame recommends $D$ immediately.

\clearpage
\section{Repetition Does Not Change the Game-Theoretic Value of Chess}

\begin{sectionoverview}
Under the threefold repetition rule in chess, a player may claim a draw when the same position occurs for the third time, and the game is drawn automatically on the fifth occurrence. We ask whether these repetition rules can affect the game-theoretic win/draw/loss outcome of chess. Keeping all other rules unchanged, including the fifty-move and seventy-five-move rules, we show that any player who can force a win can do so without repeating the same chesss position. Thus, chess has the same optimal value whether the repetition draws are eliminated or not.
\end{sectionoverview}

\subsection{Introduction}
Under the present International Chess Federation (FIDE) rules, a player may claim a draw when the same position occurs for the third time, and the game is drawn automatically on the fifth occurrence \citep[Articles.~9.2, 9.6]{FIDE2023}. We ask whether these repetition rules can affect which player can force a win under perfect play. 

We show that they cannot. More strongly, suppose that White or Black can force a win when the repetition rules are removed but all other chess rules are unchanged. That player has a winning strategy in which no FIDE position occurs twice. It implies that the theoretical win/draw/loss value is unchanged if the current repetition rules are removed or strengthened so that the second occurrence is already an automatic draw.

Kalm\'ar's classical result at first seems to settle the question. \citet[p.~79]{kalmar1929} proved that in an abstract two-person game of perfect information, a player who can force a win can do so ``without repetition'' of a position. \citet{schwalbe2001} trace this result from Zermelo's original argument through K\"onig's criticism. The important qualification is that Kalm\'ar's \emph{position} must include everything that can affect future play; identical positions must have identical continuation games \citep[p.~69]{kalmar1929}.

\citet{ewerhart2022} first recognized this state-space issue explicitly for chess. First, he showed that the value of the then-official, potentially infinite version of chess was the same as that of a finite version that declared a draw on the third occurrence of a position \citep[pp.~211--213]{ewerhart2022}. Second, and more relevant for our purpose, he introduced a ``z-position'' that records the current chess position together with the information needed for repetition and the fifty-move rule, so equal z-positions have the same continuation game \cite[pp.~213--214]{ewerhart2022}. This distinction is central to our result.

To see the distinction, let $P$ denote the FIDE position used to determine repetition. It records the placement of the pieces, the player to move, and the castling and en-passant rights relevant to the legal moves. Let $h$ be the number of individual moves since the last pawn move or capture. Then $(P,h)$ distinguishes game states, whereas the repetition rule compares only $P$. Thus, $(P,20)\ne(P,24)$ as game states, even though they are the same position for purposes of repetition. Kalm\'ar's theorem can rule out repetition of the complete state; by itself, it does not rule out a second occurrence of $P$ with a different value of $h$.

The main result takes into account this difference. For every position $P$ from which a player can force a win, consider the largest fifty-move count at which $P$ remains winning. In any winning line with no pawn move or capture, this largest winning count must strictly increase from one position to the next. But any sequence that returns to the same chess position can contain neither a capture nor a pawn move. Thus, this number must be strictly larger than itself, which is impossible.

\subsection{Results}

A \emph{position} $P$ consists of the placement of the pieces, the player to move, and the castling and en-passant information. This is the position used to determine whether the same position has occurred before under FIDE Article 9.2.3 \cite{FIDE2023}. In particular, the number of moves since the last pawn move or capture is not part of $P$.

Let $h$ be the number of individual moves (i.e., plies) since the last pawn move or capture. Under the seventy-five-move rule, a nonterminal state has $0\le h\le149$; if 150 such moves are completed without a pawn move or capture, the game is drawn automatically unless the last move is checkmate \citep[Art.~9.6.2]{FIDE2023}. Under the fifty-move rule, a draw is not automatic but a claim is possible. 

Let $G^0$ denote the standard chess without repetition rules. Let $G^2$ denote the same game except that the second occurrence of a position $P$ is an automatic draw. We write its value $V(G)$ as one of win, draw, or loss from White's perspective. We first note the following observation.

\begin{lemma}\label{lem:mono}
Fix a player $A$. If $A$ can force a win from $(P,h)$ in $G^0$, then $A$ can force a win from $(P,h')$ for every $h'\le h$.
\end{lemma}

\begin{proof}
Consider the winning strategy from $(P,h)$. Until the next pawn move or capture, the actual count beginning at $h'$ is always less than $h$. Thus a fifty-move claim or the automatic seventy-five-move draw cannot arise earlier. After a pawn move or capture, both counts return to zero and the continuations coincide. Because changing $h$ does not change the legal moves, the same strategy wins from $(P,h')$.
\end{proof}

For every position $P$ that is winning for $A$ at some value of $h$, define
\[
 H(P)=\max\{h:(P,h)\text{ is winning for }A\text{ in }G^0\}.
\]
The maximum exists because there are only finitely many nonterminal values of $h$. By Lemma~\ref{lem:mono}, if $P$ is winning at $H(P)$, it is winning at every $h\leq H(P)$.

\begin{theorem}\label{thm:main}
If a player $A$ can force a win in $G^0$ from $(P_0,h_0)$, then $A$ has a winning strategy under which no FIDE position $P$ occurs twice.
\end{theorem}

\begin{proof}
Assume that $A$ can force a win in $G^0$ from $(P_0,h_0)$. We next construct a non-repeating winning strategy. Consider any position $P$ at which $A$ is to move.  Whenever $P$ is reached, make the move that preserves the win for $A$ at state $(P,H(P))$. We first show that every nonterminal state $(P,h)$ reached under this strategy is winning for $A$ and satisfies $h\le H(P).$

Obviously, this holds at the start. Suppose it holds at some position $P$. If $A$ is to move, the chosen move is winning when the count is $H(P)$; if the opponent is to move, every move from the winning state $(P,H(P))$ must leave another state winning for $A$, or the opponent could avoid losing. The same applies to any draw claim: if the opponent could claim a draw, $(P,H(P))$ would not be winning for $A$.

If the next move is a pawn move or capture, the count returns to zero. Otherwise, it increases to $H(P)+1$. The actual count is no larger, so by Lemma~\ref{lem:mono}, the resulting state is still winning for $A$. Thus the property is preserved after every move.

Now consider any move from $P$ to $Q$, consistent with the strategy, that is neither a pawn move nor a capture. By the preceding argument, $(Q,H(P)+1)$ is winning for $A$. Hence
\begin{equation}\label{eq:increase}
H(Q)\ge H(P)+1>H(P).
\end{equation}
This holds whether the move is made by $A$ or the opponent.

To reach a contradiction, suppose that some position repeats:
\[
P_0, P_1, \cdots, P_k=P_0.
\]
No move in this sequence can be a capture or a pawn move, because any such move  is irreversible. Applying \eqref{eq:increase} gives
\[
H(P_0)<H(P_1)<\cdots<H(P_k)=H(P_0),
\]
a contradiction. Therefore no position occurs twice.

Finally, infinite play is impossible under the constructed strategy because there are finitely many FIDE positions and we ruled out repetition. Therefore, $A$ is guaranteed to win under the constructed strategy, 
\end{proof}

\begin{corollary}[Repetition rules do not change the value]\label{cor:value} $
V(G^0)=V(G^2)$. 
\end{corollary}

\begin{proof}
Suppose $A$ can force a win in $G^0$. By Theorem~\ref{thm:main}, $A$ can do so without repeating any position. Therefore a twofold, threefold, or fivefold repetition rule is never invoked, and the same strategy wins under each rule.

Conversely, suppose $A$ can force a win in $V(G^2)$. A winning strategy cannot allow a repetition draw. Thus, removing the repetition draw does not change the winning strategy. As a result, the winning positions are the same under all these rules. The remaining positions are draws, so the win/draw/loss value is unchanged.
\end{proof}

\begin{remark}[Arbitrary repetition and move-count thresholds]
The particular numbers three, five, fifty, and seventy-five are not essential. Fix any finite $X$-move rule of the same form, in which a pawn move or capture resets the count, and let $G_X^{n}$ denote the corresponding game in which the $n$th occurrence of the same chess position is a draw. Then, for every fixed $X$ and every $n\ge2$, $V(G_X^{(0)})=V(G_X^{n})$.
The proof remains the same: $H(P)$ is defined for the chosen $X$, and there is a winning strategy without a second repetition. Thus, no numerical relation between $X$ and $n$ is needed for the result. Note however that this does not imply that changing $X$ itself leaves the value of chess unchanged.\footnote{More generally, the value-equivalence result holds for any finite two-player deterministic perfect-information game with win/draw/loss outcomes in which the complete state is $(p,c)$, actions and their base successors are independent of the bounded counter $c$, every continuing action either increases $c$ or resets it to zero, and no closed path of base positions contains a resetting action. If repetition is defined on the base position $p$, then, for every $c$ and every $n\ge2$, adding or removing a draw on the $n$th occurrence of $p$ does not affect the value.}
\end{remark}

\clearpage
\phantomsection

\end{document}